\documentclass[]{elsart}

\usepackage[utf8]{inputenc}
\usepackage{amssymb,amsmath}
\usepackage{epsfig}
\usepackage{color}
\usepackage{srcltx}
\usepackage{subfigure}
\usepackage{xspace}
\usepackage{enumerate}
\usepackage{array}
\usepackage[mathlines]{lineno}

\newtheorem{theorem}{Theorem}
\newtheorem{lemma}[theorem]{Lemma}
\newtheorem{corollary}[theorem]{Corollary}
\newtheorem{proposition}[theorem]{Proposition}

\newtheorem{observation}[theorem]{Observation}
\newtheorem{openprob}[theorem]{Open Problem}

\newenvironment{proof}{\par \noindent \textbf{Proof.}
}{\hfill$\Box$\medskip}
\newenvironment{proofof}[1]{\par \noindent \textbf{Proof of #1.}
}{\hfill$\Box$\medskip} 

\newcommand{\set}[1]{\left\{#1\right\}}
\newcommand{\paren}[1]{\left(#1\right)}
\newcommand{\ceil}[1]{\left\lceil #1 \right\rceil}
\newcommand{\fq}{\mathbb{F}_q}
\newcommand{\spk}[1]{SP_{#1}}
\newcommand{\spq}{\spk{q}}

\newcommand{\tr}{Tr(\spq)}
\newcommand{\trk}[1]{Tr(\spk{#1})}
\newcommand{\alphab}{\bar{\alpha}}
\newcommand{\Prop}[1]{P_{#1}}
\newcommand{\vt}[1]{\overrightarrow{#1}}
\newcommand{\resign}[2]{#1^{(#2)}}

\DeclareMathOperator{\sq}{sq}
\DeclareMathOperator{\twin}{atw}

\date{}

\begin{document}

\begin{frontmatter}

\title{Homomorphisms of signed planar graphs\thanksref{grant}} 
\author{Pascal Ochem},
\ead{Pascal.Ochem@lirmm.fr}
\ead[url]{http://www.lirmm.fr/\~{}ochem}
\author{Alexandre Pinlou\thanksref{UPV}}
\ead{Alexandre.Pinlou@lirmm.fr}
\ead[url]{http://www.lirmm.fr/\~{}pinlou}

\thanks[grant]{This work was partially supported by the ANR grant
 EGOS 12-JS02-002-01 and by the PEPS grant HOGRASI.}
\thanks[UPV]{Second affiliation: D\'epartement de Math\'ematiques et Informatique Appliqu\'es,
 Universit\'e Paul-Val\'ery, Montpellier 3, France.}

\address{LIRMM, Universit\'e Montpellier 2, CNRS, France.}

\author{Sagnik Sen}
\ead{Sagnik.Sen@labri.fr}

\address{LaBRI, Universit\'e de Bordeaux, CNRS, France.}

\begin{abstract}
  Signed graphs are studied since the middle of the last century.
  Recently, the notion of homomorphism of signed graphs has been
  introduced since this notion captures a number of well known
  conjectures which can be reformulated using the definitions of
  signed homomorphism.
  
  In this paper, we introduce and study the properties of some target
  graphs for signed homomorphism. Using these properties, we
  obtain upper bounds on the signed chromatic numbers of graphs with
  bounded acyclic chromatic number and of signed planar graphs with
  given girth.
\end{abstract}


\end{frontmatter}

\journal{}

\section{Introduction}

The class of signed graphs is a natural graph class where the edges
are either positive or negative. They were first introduced to handle
problems in social psychology: positive edges link friends whereas
negative ones link enemies. 

In the area of graph theory, they have been used as a
way of extending classical results in graph coloring such as
Hadwiger’s conjecture. Guenin~\cite{gue05} introduced the notion of
signed homomorphism for its relation with a well known conjecture of
Seymour. In 2012, this notion has been further
developed by Naserasr et al.~\cite{nrs12} as this theory captures a
number of well known conjectures which can be reformulated using the
definitions of signed homomorphism. In this paper, we study signed
homomorphisms for themselves. 

\bigskip

A \emph{signified graph} $(G,\Sigma)$ is a graph $G$ with an
assignment of positive ($+1$) and negative ($-1$) signs to its edges
where $\Sigma $ is the set of negative edges. In all the figures,
negative edges are drawn with dashed edges. Figure~\ref{fig:example}
gives an example of signified graph. \emph{Resigning} a vertex $v$ of
a signified graph $(G,\Sigma)$ corresponds to give the opposite sign
to the edges incident to $v$. Given a signified graph $(G,\Sigma)$ and
a set of vertices $X\subseteq V(G)$, the graph obtained from
$(G,\Sigma)$ by resigning every vertex of $X$ is denoted by
$(G,\resign{\Sigma}{X})$.

\begin{figure}
  \begin{center}
    \subfigure[\label{fig:example}]{\includegraphics[scale=0.8]{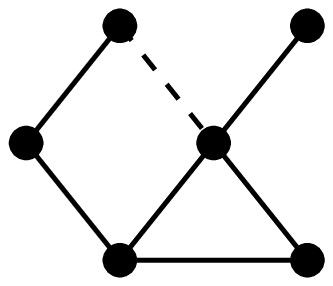}} $\qquad\qquad$
    \subfigure[]{\includegraphics[scale=0.8]{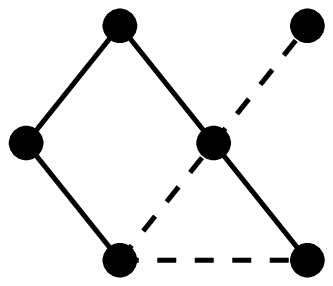}}
  \caption{Two equivalent signified graphs.\label{fig:equivalent}}
  \end{center}
\end{figure}

Two signified graphs $(G,\Sigma_1)$ and $(G,\Sigma_2)$ are said to be
\emph{equivalent} if we can obtain $(G,\Sigma_1)$ from $(G,\Sigma_2)$
by resigning some vertices of $(G,\Sigma_2)$, i.e.,
$\Sigma_2=\resign{\Sigma_1}{X}$ for some $X\subseteq V(G)$; in such a
case, we use the notation $(G, \Sigma_1)\sim(G,\Sigma_2)$ (see
Figure~\ref{fig:equivalent} for an example of equivalent signified
graphs).  Each equivalence class defined by the resigning process is
called a \emph{signed graph} and can be denoted by any member of its
class.  We might simply use $(G)$ for a signified/signed graph when
its set of negative edges is clear from the context, while $G$ refers
to its underlying unsigned graph.

An \emph{$m$-edge-colored graph} $G$ is a graph where the vertices are
linked by edges $E(G)$ of $m$ types. In other words, there is a
partition $E(G)=E_{1}(G)\cup \ldots \cup E_{m}(G)$ of the edges of
$G$, where $E_{j}(G)$ contains all edges of type $j$.  Since signified
graphs are defined with two types of edges (i.e.  positive and
negative edges), they correspond to $2$-edge-colored graphs.  In this
paper, known and new results on $2$-edge-colored graphs are stated in
terms of signified graphs.
 
Given two graphs $(G,\Sigma)$ and $(H,\Lambda)$, $\varphi$ is a
\emph{signified homomorphism} of $(G,\Sigma)$ to $(H,\Lambda)$ if
$\varphi : V(G) \longrightarrow V(H)$ is a mapping such that every
edge of $(G,\Sigma)$ is mapped to an edge of the same sign of $(H,
\Lambda)$. Given two graphs $(G,\Sigma_1)$ and $(H,\Lambda_1)$, we
say that there is a \emph{signed homomorphism} $\varphi$ of $(G,\Sigma_1)$
to $(H,\Lambda_1)$ if there exists $(G,\Sigma_2)\sim (G,\Sigma_1)$ and
$(H,\Lambda_2) \sim (H,\Lambda_1)$ such that $\varphi$ is a signified
homomorphism of $(G,\Sigma_2)$ to $(H,\Lambda_2)$.
 
\begin{lemma}\label{altdef}
  If $(G,\Sigma)$ admits a signed homomorphism to $(H,\Lambda)$, then
  there exists $(G,\Sigma')\sim (G,\Sigma)$ such that $(G,\Sigma')$
  admits a signified homomorphism to $(H,\Lambda)$.
\end{lemma}
 
\begin{proof}
  Since $(G,\Sigma)$ admits a signed homomorphism to $(H,\Lambda)$,
  this implies that there exist $(G,\Sigma'')\sim(G,\Sigma)$,
  $(H,\Lambda')\sim(H,\Lambda)$ and a signified homomorphism $\varphi$
  of $(G,\Sigma'')$ to $(H,\Lambda')$. Let $X\subseteq V(H)$ be the
  subset of vertices of $H$ such that
  $(H,\Lambda)=(H,\resign{\Lambda'}{X})$. Now let $Y=\set{v\in V(G)
    \mid \varphi(v)\in X}$. Let $(G,\Sigma') =
  (G,\resign{\Sigma''}{Y})$; it is clear that $\varphi$ is a signified
  homomorphism of $(G,\Sigma')$ to $(H,\Lambda)$.
\end{proof} 
 
As a consequence of the above lemma, when dealing with signed
homomorphisms, we will not need to resign the target graph.

The \emph{signified chromatic number} $\chi_2(G,\Sigma)$ of the graph
$(G,\Sigma)$ is the minimum order (number of vertices) of a graph
$(H,\Lambda)$ such that $(G,\Sigma)$ admits a signified homomorphism
to $(H,\Lambda)$. Similarly, the \emph{signed chromatic number}
$\chi_s(G,\Sigma)$ of the graph $(G,\Sigma)$ is the minimum order of a
graph $(H,\Lambda)$ such that $(G,\Sigma)$ admits a signed
homomorphism to $(H,\Lambda)$; equivalently,
$\chi_s(G,\Sigma)=\min\set{\chi_2(G,\Sigma') \mid
  (G,\Sigma')\sim(G,\Sigma)}$.

The \emph{signified chromatic number $\chi_2(G)$} of a graph $G$ is
defined as $\chi_2(G)=\max\{\chi_2(G,\Sigma) \mid \Sigma\subseteq E(G)
\}$. The \emph{signified chromatic number $\chi_2(\mathcal{F})$} of a
graph class $\mathcal{F}$ is defined as
$\chi_2(\mathcal{F})=\max\{\chi_2(G)\mid G\in\mathcal{F}\}$.  The
\emph{signed chromatic numbers} of a graph and a graph class are
defined similarly.

Another equivalent definition of the signified chromatic numbers can
be given by defining the signified coloring. A \emph{signified
  coloring} of a signified graph $(G)$ is a proper vertex-coloring
$\varphi$ of $G$ such that if there exist two edges $uv$ and $xy$ with
$\varphi(u)=\varphi(x)$ and $\varphi(v)=\varphi(y)$, then these two
edges have the same sign. Hence, the \emph{signified chromatic number}
of the signified graph $(G)$ is the minimum number of colors needed
for a signified coloring of $(G)$.

In this paper, we studied signified and signed homomorphisms of
outerplanar and planar graphs of given girth. The paper is organized
as follows. We introduce the notation in Section~\ref{sec:notation}.
Section~\ref{sec:target} is devoted to introduce and study the
properties of several families of target graphs, namely the
\emph{Anti-twinned graph $AT(G,\Sigma)$}, the \emph{signified Zielonka
  graph $ZS_k$}, the \emph{signified Paley graph $\spq$}, and the
\emph{signified Tromp Paley graph $\tr$}. We study signified
homomorphisms of planar graphs (resp. outerplanar graphs) in
Section~\ref{sec:signified} and we provide lower and upper bounds on
the signified chromatic number. We get upper bounds on the signed
chromatic number of planar graphs (resp. outerplanar graphs) of given
girth in Section~\ref{sec:signed}. We finally conclude in
Section~\ref{sec:conclusion}.

\section{Notations}\label{sec:notation}

For a vertex $v$ of a signified graph $(G)$, $d_{(G)}(v)$ denotes the
degree of $v$. The set of positive neighbors of $v$ is denoted by
$N^+_{(G)}(v)$ and the set of negative neighbors of $v$ is denoted by
$N^-_{(G)}(v)$. Thus, the set of neighbors of $v$, denoted by
$N_{(G)}(v)$, is $N_{(G)}(v)=N^+_{(G)}(v)\cup N^-_{(G)}(v)$. A vertex
of degree $k$ (resp. at least $k$, at most $k$) is called a $k$-vertex
(resp. $^\ge k$-vertex, $^\le k$-vertex).  If a vertex $u$ is adjacent
to a $k$-vertex ($^\ge k$-vertex, $^\le k$-vertex) $v$, then $v$ is a
$k$-neighbor (resp. $^\ge k$-neighbor, $^\le k$-neighbor) of $u$. A
path of length $k$ (i.e. formed by $k$ edges) is called a $k$-path.
Given a planar graph $G$ with its embedding in the plane and a vertex
$v$ of $G$, we say that a sequence $(u_1, u_2,\cdots, u_k)$ of
neighbors of $v$ are \emph{consecutive} if $u_1, u_2,\cdots, u_k$
appear consecutively around $v$ in $G$ (clockwise or
counterclockwise).

\section{Target graphs}\label{sec:target}

In this section, our goal is not only to find target graphs that will
give the required upper bounds of our results of
Sections~\ref{sec:signified} and~\ref{sec:signed}. We aim at
describing several families of target graphs that may be useful for
signified and signed homomorphisms and we determine their properties.
To this end, we describe below the \emph{Anti-twinned graph}
construction, the \emph{signified Zielonka graph $ZS_k$}, the
\emph{signified Paley graph} and the \emph{Tromp signified Paley
  graph}.

\subsection{The Anti-twinned graph}

Let $(G,\Sigma)$ be a signified graph and let $(G^0,\Sigma^0)$ and
$(G^1,\Sigma^1)$ be two isomorphic copies of $(G,\Sigma)$. In the
following, given a vertex $u\in V(G)$, we denote $u_i$ the
corresponding vertex of $u$ in the isomorphic copy $(G^i,\Sigma^i)$
of $(G,\Sigma)$. We define the \emph{anti-twinned graph} $AT(G,\Sigma)=
(H,\Lambda)$ on $2|V(G)|$ vertices as follows:

\begin{itemize}
\item $V(H)=V(G^0) \cup V(G^1)$
\item $E(H)=E(G^0) \cup E(G^1) \cup \set{u_iv_{1-i}: uv\in E(G)}$
\item $\Lambda=\Sigma^0 \cup \Sigma^1 \cup \set{u_iv_{1-i}: uv\in E(G)
    \setminus \Sigma}$
\end{itemize}

\begin{figure}
 \begin{center} 
 \includegraphics[scale=0.7]{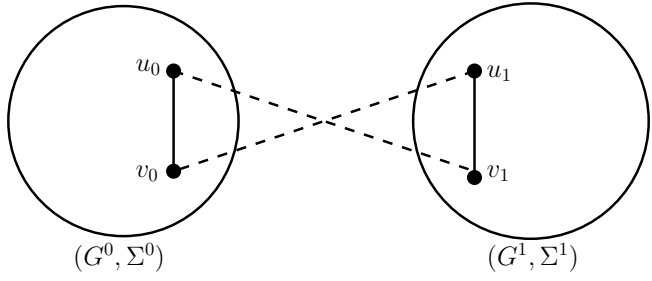}
 \caption{The anti-twinned graph $AT(G,\Sigma)$.\label{fig:doubled_graph}}
 \end{center}
\end{figure}
 
Figure~\ref{fig:doubled_graph} illustrates the construction of
$AT(G,\Sigma)$. We can observe that for every $u_i\in V(G^i)$,
there is no edge between $u_i$ and $u_{1-i}$. By
construction we have the following property: 
\begin{linenomath}
 $$\forall u_i\in AT(G,\Sigma) : N^+(u_i)=N^-(u_{1-i}) \mbox{ and } N^-(u_i)=N^+(u_{1-i})$$
\end{linenomath}
Such pairs of vertices
are called \emph{anti-twin vertices}, and for any $u\in AT(G,\Sigma)$
we denote by $\twin(u)$ the anti-twin vertex of $u$. Remark that
$\twin(\twin(u))=u$. This notion can be extended to sets in a
standard way: for a given $W\subseteq V(G^i)$,
$W=\{v_1,v_2,\ldots,v_k\}$, then $\twin(W)=
\set{\twin(v_1),\twin(v_2),\ldots,\twin(v_k)}$.

We say that a signified graph is \emph{anti-twinned} if it is the
anti-twinned graph of some signified graph.

\begin{observation}\label{obs:doubled_twin}
A signified graph is anti-twinned if and only if each of its vertices has a
unique anti-twin.
\end{observation}

\begin{lemma}\label{lem:signed-signified}
  A graph $(G,\Sigma)$ admits a signed homomorphism to $(H,\Lambda)$ if
  and only if $(G,\Sigma)$ admits a signified homomorphism
  to $AT(H,\Lambda)$.
\end{lemma}

\begin{proof}
 Let $\varphi$ be a signed homomorphism of $(G,\Sigma)$ to $(H,\Lambda)$.
 This implies that $\varphi$ is a signified homomorphism of $(G,\Sigma_1)$
 to $(H,\Lambda)$, where $(G,\Sigma_1) \sim (G,\Sigma)$. 
 
 Let $X\subseteq V(G)$ be the subset of vertices of $G$ such that
 $(G,\Sigma)=(G,\resign{\Sigma_1}{X})$. By definition of
 the resigning process, only the edges of the edge-cut
 between $V(G)\setminus X$ and $X$ get their sign changed.

 Let $\varphi' : V(G) \to V(AT(H))$ be defined as follows:
 \begin{equation}
  \varphi'(u)=\begin{cases}
   \twin(\varphi(u)), & \text{if $u\in X$},\\
   \varphi(u), & \text{otherwise}.
  \end{cases}\nonumber
 \end{equation}

 By construction of $AT(H,\Lambda)$, if $u$ and $v$ induce an
 edge of sign $s$, then $u$ and $\twin(v)$ induce an edge of
 sign $-s$. Therefore, it is easy to see that $\varphi'$ is a
 signified homomorphism of $(G,\Sigma)$ to $AT(H,\Lambda)$.
 This proves the only if part.
 
 For the if part, suppose that $(G,\Sigma)$ admits a signified
 homomorphism $\psi$ to $AT(H,\Lambda)$. The signified graph
 $AT(H,\Lambda)$ is obtained from two isomorphic copies $(H^0,\Lambda^0)$
 and $(H^1,\Lambda^1)$ of $(H,\Lambda)$. Let $Y\subseteq V(G)$ be the subset of
 vertices of $G$ such that, for all $y\in Y$, $\psi(y)$ is a vertex of
 $H^1$. Let $(G,\Sigma_1)=(G,\resign{\Sigma}{Y})$. 
 Let $\psi' : V(G) \to V(AT(H))$ be defined as follows:
 \begin{equation}
  \psi'(u)=\begin{cases}
   \twin(\psi(u)), & \text{if $u\in Y$},\\
   \psi(u), & \text{otherwise}.
  \end{cases}\nonumber
 \end{equation}
 It is easy to see that $\psi'$ is a signified homomorphism of
 $(G,\Sigma_1)$ to $AT(H,\Lambda)$ such that every vertex maps to a vertex
 of $H^0$. Then $\psi'$ is a signified homomorphism from
 $(G,\Sigma_1)$ to $(H,\Lambda)$ and thus $(G,\Sigma)$ admits a signed
 homomorphism to $(H,\Lambda)$.
\end{proof}

\begin{corollary}\label{cor:bound_signed}
  If $(G,\Sigma)$ admits a signified homomorphism to an anti-twinned
  graph $T$, then we have:
 \begin{enumerate}
 \item\label{cor:bound_signed-1} $\chi_s(G)\le\frac{|V(T)|}{2}$.
 \item\label{cor:bound_signed-2} $(G,\Sigma')$ admits a signified homomorphism to $T$ for every $(G,\Sigma') \sim (G,\Sigma)$.
 \end{enumerate}
\end{corollary}

\subsection{The signified Zielonka graph $ZS_k$}\label{subsec:ask}

The Zielonka graph $Z_k$ was introduced by Zielonka~\cite{z90} in the
theory of bounded timestamp systems. Alon and Marshall~\cite{am98}
adapted this construction to signified graphs to obtain the
\emph{signified Zielonka graph} $ZS_k$. They used this graph to get
bounds on the signified chromatic number of graphs that admits an
acyclic $k$-coloring.

Let us describe the construction of the signified Zielonka graph
$ZS_k$.  Every vertex is of the form
$(i;\alpha_1,\alpha_2,\ldots,\alpha_k)$ where $1\le i\le k$,
$\alpha_j\in\set{+1,-1}$ for $j\ne i$ and $\alpha_i=0$. There are
clearly $k\cdot 2^{k-1}$ vertices in this graph. For $i\ne j$, there
is an edge between the vertices
$(i;\alpha_1,\alpha_2,\ldots,\alpha_k)$ and
$(j;\beta_1,\beta_2,\ldots,\beta_k)$ and the sign of this edge is
given by the product $\alpha_j\times \beta_i$.

\begin{proposition}\label{prop:ZSk}
 The graph $ZS_k$ is anti-twinned.
\end{proposition}

\begin{proof}
 By Observation~\ref{obs:doubled_twin}, we have to show that every
 vertex has an anti-twin.
 We claim that the anti-twin of the vertex $v=(i;\alpha_1,\alpha_2,\ldots,\alpha_k)$ is
 the vertex $v'=(i;-\alpha_1,-\alpha_2,\ldots,-\alpha_k)$. Indeed,
 $v$ and $v'$ are not adjacent and it is easy to check that for every
 edge $uv$, the edge $uv'$ exists, and that $uv$ and $uv'$ have opposite signs.
\end{proof}

\subsection{The signified Paley graph $\spq$}

\begin{figure}
\begin{center}
 \includegraphics{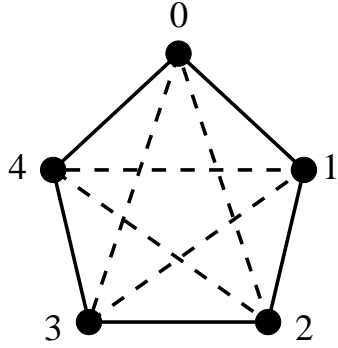}
 \caption{The signified graph $\spk{5}$.\label{fig:SP5}}
\end{center}
\end{figure}

In the remaining, $q$ is any prime power such that $q\equiv 1\pmod 4$.
There is a unique (up to isomorphism) finite field $\fq$ of order $q$.
Let $g$ be a generator of the field $\fq^*$. For every $v\in\fq^*$,
let $\sq : \fq^* \to \set{-1,+1}$ be the function \emph{square}
defined as $\sq(v)=+1$ if $v$ is a square and $\sq(v)=-1$ if $v$ is a
non-square. Note that $\sq(g^t)=(-1)^t$ since $g$ is necessarily a non-square.

The \emph{Paley graph} $P_q$ is the undirected graph with vertex set
$V(P_q)=\fq$ and edge set $E(P_q)=\lbrace xy\mid\sq(y-x)=+1\}$. Since
$-1$ is a square in $\fq$, $\sq(x-y)=\sq(y-x)$ and therefore the
definition of an edge is consistent. We also know that a Paley graph
is self-complementary~\cite{s62} and edge-transitive.

A $k$-regular graph $G$ with $n$ vertices is said to be \emph{strongly regular}
if (1) every two adjacent vertices have $\lambda$ common neighbors and
(2) every two non-adjacent vertices have $\mu$ common neighbors.
Such a graph is said to be a strongly regular graph with parameters $(n,k,\lambda,\mu)$.
Paley graphs $P_q$ are known to be strongly regular graphs with parameters $(q,\frac{q-1}{2}, \frac{q-5}{4},\frac{q-1}{4})$.

For any prime power $q \equiv 1 \pmod 4$, we define the
\emph{signified Paley graph} $\spq=(K_q,\Sigma)$ as the complete graph
on $q$ vertices with $V(\spq)=\fq$ and
$\Sigma=\set{xy\mid\sq(y-x)=-1}$. That is, $\spq$ is obtained from
the Paley graph $P_q$ by replacing the non-edges by negative edges.
Figure~\ref{fig:SP5} represents the signified Paley graph $\spk{5}$.
Since $P_q$ is edge-transitive and self-complementary, $\spq$ is
clearly edge-transitive.

\subsection{The Tromp signified Paley graph $Tr(\spq)$}

Given an oriented graph $\vt{G}$, Tromp~\cite{t} proposed a
construction of an oriented graph $Tr(G)$ called Tromp graph. We adapt
this construction to signified graphs as follows. 

For a given graph $(G,\Sigma)$, let us denote $(G^+,\Sigma)$ the graph
obtained from $(G,\Sigma)$ by adding a universal vertex positively
linked to all the vertices of $(G,\Sigma)$. 

\begin{figure}
\begin{center}
 \includegraphics[scale=2]{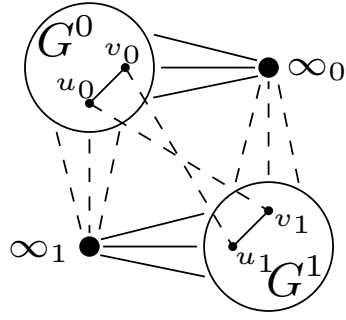}
 \caption{The signified Tromp graph $Tr(G,\Sigma)$.\label{fig:tromp}}
\end{center}
\end{figure}

Then, the \emph{Tromp signified graph} $Tr(G,\Sigma)$ of
$(G,\Sigma)$ is defined to be the anti-twinned graph of $(G^+,\Sigma)$,
that is $Tr(G,\Sigma) \cong AT(G^+,\Sigma)$. Figure~\ref{fig:tromp}
illustrates the construction of $Tr(G,\Sigma)$.

By construction, $Tr(G,\Sigma)$ is obtained from two isomorphic copies
$(G^0,\Sigma^0)$ and $(G^1,\Sigma^1)$ of $(G,\Sigma)$ plus $2$
vertices $\infty_0$ and $\infty_1$.  

In the remainder, we focus on the specific graph family obtained by applying the Tromp's construction
to the signified Paley graph $\spq$. 

We consider the Tromp signified Paley graph $\tr$ on $2q+2$ vertices
obtained from $\spq$.  In the remainder of this paper, the vertex set
of $\tr$ is
$V(\tr)=\{0_0,1_0,\ldots,q-1_0,\infty_0,0_1,1_1,\ldots,q-1_1,\infty_1\}$
where $\{0_i,1_i,\ldots,$ $q-1_i\}$ is the vertex set of the isomorphic
copy $\spq^i$ of $\spq$ ($i\in\set{0,1}$); thus, for every
$u_i\in\{0_i,1_i,\ldots,q-1_i,\infty_i\}$, we have $\twin(u_i)=u_{1-i}$.
In addition, for every $u\in V(\tr)$, we have by construction
$|N^+_{\tr}(u)|=|N^-_{\tr}(u)|=q$.

Let $i,j\in\set{0,1}$ and $u,v\in\fq$. If $i=j$, then $u_i$ and $v_j$
are in the same isomorphic copy of $\spq$ in $\tr$; in this case, the
sign of the edge $u_iv_j$ is $\sq(u-v)$ by definition of $\spq$.
If $i\ne j$, then $u_i$ and $v_j$ are in distinct isomorphic copies of $\spq$ in $\tr$;
in this case, the sign of the edge $u_iv_j$ is $-\sq(u-v)$.
Therefore, in both cases, the sign of the edge $u_iv_j$ is

\begin{equation}\label{eq:sign}
\sq(u-v)\times (-1)^{i+j}.
\end{equation}

Let $i,j\in\set{0,1}$ and $v\in\fq$. If $i=j$, then the sign of the edge
$\infty_iv_j$ is $+1$, while it is $-1$ when $i\ne j$. Therefore, in both
case, the sign of the edge $\infty_iv_j$ is

\begin{equation}\label{eq:signinfty}
(-1)^{i+j}.
\end{equation}

The graph $\tr$ has remarkable symmetry and some useful properties given below.

\begin{lemma}\label{lem:transitivity}
The signified graph $\tr$ is vertex-transitive.
\end{lemma}

\begin{proof}
The mapping $\gamma_1 : V(\tr) \to V(\tr)$ defined as $\gamma_1(u_i)=u_{1-i}$ is clearly an automorphism of $\tr$.

Recall first that $\spq$ is edge-transitive and so vertex-transitive.
If $\varphi$ is an automorphism of $\spq$, we can define the corresponding automorphism $\gamma_2$ of $\tr$ as:
$$\gamma_2 : u_i \to 
 \left\{
  \begin{array}{ll} 
    u_i & \mbox{if $u=\infty$} \\
    (\varphi(u))_i & \mbox{if $u\ne\infty$}
  \end{array} 
 \right.$$

We eventually define the mapping $\gamma_3 : V(\tr) \to V(\tr)$ as:
$$\gamma_3 : u_i \to 
 \left\{
   \begin{array}{ll} 
    \infty_i & \mbox{if $u=0$} \\
     0_i & \mbox{if $u=\infty$} \\
     (u^{-1})_i & \mbox{if $u$ is a non-zero square} \\
     (u^{-1})_{1-i} & \mbox{if $u$ is a non square} \\
   \end{array} 
 \right.$$

Let $S$ be the sign of the edge $u_iv_j$. To prove that $\gamma_3$ is an automorphism
of $\tr$, we will show that $\gamma_3$ maps $u_iv_j$ to an edge of sign $S'=S$.

Let $g$ be a generator of the field $\fq$. Any vertex $v$ of $\spq$ is
an element of $\fq$ and therefore $v=g^t$ for some $t$ when $v\ne 0$.
Recall that $\sq(g^t)=(-1)^t$.

\begin{itemize}
\item When $u,v\in V(\tr)$ are neither $0$ nor $\infty$, we have
$u_iv_j=(g^t)_i(g^{t'})_j$ for some $t$ and $t'$. If $t$ is even, that
is $g^t$ is a square, then $\gamma_3((g^t)_i)=(g^{-t})_i$; otherwise,
$t$ is odd and thus $\gamma_3((g^t)_i)=(g^{-t})_{1-i}$. Let $i'=i+t
\pmod 2$ and $j'=j+t'\pmod 2$. Clearly, if $t$ is even (resp. $t'$)
is even, then $i'=i$ (resp. $j'=j$); otherwise we have $i'=1-i$
(resp. $j'=1-j$). Therefore, the function $\gamma_3$ maps the edge
$(g^t)_i(g^{t'})_j$ to the edge $(g^{-t})_{i'}(g^{-t'})_{j'}$. Now,
let us check that the sign of these two edges is the same.

By Equation~(\ref{eq:sign}), the sign of the edge $u_iv_j$ is
$S=\sq(g^t-g^{t'})\times (-1)^{i+j}$ and the sign of the edge
$(g^{-t})_{i'}(g^{-t'})_{j'}$ is $S'=\sq(g^{-t}-g^{-t'})\times
(-1)^{i'+j'}$. We then have: 

\begin{align*}
S &=\sq(g^t-g^{t'})\times (-1)^{i+j}\\
&=\sq(g^{t'}-g^{t})\times (-1)^{i+j}\\
&=\sq((g^{-t}-g^{-t'})\times g^{t+t'})\times (-1)^{i+j}\\
&=\sq(g^{-t}-g^{-t'})\times \sq(g^{t+t'})\times (-1)^{i+j}\\
&=\sq(g^{-t}-g^{-t'})\times (-1)^{t+t'}\times (-1)^{i+j}\\
&=\sq(g^{-t}-g^{-t'})\times (-1)^{i+t+j+t'}\\
&=\sq(g^{-t}-g^{-t'})\times (-1)^{i'+j'}=S'\\
\end{align*}

\item If $u=0$ and $v=\infty$, it is clear that the edge $u_iv_j$ maps to an
edge of the same sign.

\item Consider now the case $u=0$ and $v\not\in\set{0,\infty}$.
Let $v=g^t$ for some $t$. The function $\gamma_3$ maps the edge
$0_i(g^t)_j$ on $\infty_i(g^{-t})_{j'}$ where $j'=j+t\pmod 2$.
By Equations~\ref{eq:sign} and~\ref{eq:signinfty}, the sign of the edge $0_i(g^t)_j$ is $S=\sq(g^t)\times (-1)^{i+j}$
and the sign of the edge $\infty_i(g^-t)_{j'}$ is $S'=(-1)^{i+j'}$. We then have:

$$ S=\sq(g^t)\times(-1)^{i+j}=(-1)^t\times(-1)^{i+j}=(-1)^{i+j+t}=(-1)^{i+j'}=S' $$

\item The case $u=\infty$ and $v\not\in\set{0,\infty}$ is similar to
  the previous case.
\end{itemize}

Therefore, the function $\gamma_3$ maps any edge to an edge of the same sign.

Combining the automorphisms $\gamma_1$, $\gamma_2$ and $\gamma_3$
easily proves that $\tr$ is vertex-transitive.
\end{proof}

We define an \emph{anti-automorphism} of a signified graph $(G,\Sigma)$
as a permutation $\rho$ of the vertex set $V(G)$ such that $uv$ is
a positive (resp. negative) edge if and only if $\rho(u)\rho(v)$ is a
negative (resp. positive) edge.

\begin{lemma}\label{lem:antiauto}
The graph $\tr$ admits an anti-automorphism.
\end{lemma}

\begin{proof}
 Let $n$ be any non-square of $\fq$. We define the mapping $\gamma_n
 : V(\tr) \to V(\tr)$ as:
 $$\gamma_n : u_i \to\left\{
  \begin{array}{ll} 
   u_i & \mbox{if } u=\infty \\
   (n\times u)_{1-i} & \mbox{if } u\ne\infty \\
  \end{array} 
 \right.$$
 Let us check that $\gamma_n$ maps every edge $u_iv_j\in E(\tr)$ to an
 edge of opposite sign. 
 
 Let $u,v\ne\infty$. By definition, $\gamma_n$ maps $u_iv_j$ to
 $(n\times u)_{1-i}(n\times v)_{1-j}$. By Equations~\ref{eq:sign}
 and~\ref{eq:signinfty}, the sign of the edge $u_iv_j$ is
 $S=\sq(v-u)\times (-1)^{i+j}$ and the sign of the edge $(n\times
 u)_{1-i}(n\times v)_{1-j}$ is $S'=\sq(n\times(v-u))\times
 (-1)^{1-i+1-j}$.
 \begin{align*}
  S' &=\sq(n\times(v-u))\times (-1)^{1-i+1-j}\\
  &=\sq(n)\times\sq(v-u)\times (-1)^{i+j}\\
  &=-\sq(v-u)\times(-1)^{i+j}=-S\\
 \end{align*}
 Now, let $u=\infty$ and $v\ne\infty$. The mapping $\gamma_n$ maps
 $\infty_iv_j$ to $\infty_{i}(n\times v)_{1-j}$.
 By Equation~\ref{eq:signinfty}, the sign of the edge $\infty_iv_j$ is $S=(-1)^{i+j}$
 and the sign of the edge $\infty_{i}(n\times v)_{1-j}$ is $S'=(-1)^{i+1-j}=-S$.
\end{proof}

\begin{lemma}\label{lem:triangle-transitif}
 If there exists an isomorphism $\psi$ that maps the triangle
 $(u_i,v_j,w_k)$ to the triangle $(u'_{i'},v'_{j'},w'_{k'})$, then $\psi$
 can be extended to an automorphism of $\tr$.
\end{lemma}

\begin{proof}
  There exists four types of triangles depending on the sign of their
  edges.

 \begin{itemize}
 \item Let us first consider the triangles $(u_i,v_j,w_k)$ and
   $(u'_{i'},v'_{j'},w'_{k'})$ with 3 positive edges. To prove that
   $\psi$ can be extended to an automorphism of $\tr$, it suffices to
   prove that for every triangle $(u_i,v_j,w_k)$, there exists an
   automorphism $\psi'$ that maps $(u_i,v_j,w_k)$ to
   $(0_0,1_0,\infty_0)$.  Using the vertex transitivity of $\tr$
   (Lemma~\ref{lem:transitivity}), there exists an automorphism
   $\varphi$ that maps $w_k$ to $\infty_0$.  Then, since all positive
   edges incident to $\infty_0$ have their extremities in $\spq^0$,
   $\varphi$ necessarily maps the edge $u_iv_j$ to an edge $u'_0v'_0$
   in $\spq^0$.  Since $\spq$ is edge-transitive, we can finally map
   $u'_0v'_0$ to $0_01_0$.

 \item Consider now the triangles $(u_j,v_j,w_k)$ and
  $(u'_{j'},v'_{j'},w'_{k'})$ with 3 negative edges. 
  Let $\overline\tr$ be the signified graph obtained from $\tr$ by
  changing the sign of every edge. By Lemma~\ref{lem:antiauto},
  $\overline\tr$ is isomorphic to $\tr$. By the previous
  item, there exists an automorphism that maps $(u_j,v_j,w_k)$ to
  $(u'_{j'},v'_{j'},w'_{k'})$  in $\overline\tr$, and thus in $\tr$.
    
 \item Finally consider the triangles $(u_i,v_j,w_k)$ and
  $(u'_{i'},v'_{j'},w'_{k'})$ with one edge of sign $S$ and 2 edges
  of sign $-S$. Let $u_i$ and $u'_{i'}$ be the vertex incident to
  the edges of sign $-S$. Consider the triangles
  $(\twin(u_i),v_j,w_k)$ and $(\twin(u'_{i'}),v'_{j'},$ $w'_{k'})$; they
  have 3 edges of sign $S$. By the two previous cases, there exists
  an automorphism $\psi$ that maps $(\twin(u_j),v_j,w_k)$ to
  $(\twin(u'_{j'}),v'_{j'},w'_{k'})$. Since $\psi$ preserves 
  anti-twinning, $\psi$ also maps $u_i$ to $u'_{i'}$.
 \end{itemize}
\end{proof}

\subsection{Coloring properties of target graphs}

A \emph{signed vector} of size $k$ is a $k$-tuple
$\alpha=(\alpha_1,\alpha_2,\ldots,\alpha_k)\in \{+1,-1\}^k$.
For a given signed vector $\alpha$, its \emph{conjugate} is the $k$-tuple $\alphab=(-\alpha_1,-\alpha_2,\ldots,-\alpha_k)$.

Given a sequence of $k$ distinct vertices $X_k=(v_1,v_2,\ldots,v_k)$
of a signified graph $(G,\Sigma)$ that induces a clique, a vertex $u\in V(G)$
is an \emph{$\alpha$-successor} of $X_k$ if, for every $i\in\set{1,2,\ldots,k}$,
the sign of the edge $uv_i$ is $\alpha_i$. 
The set of $\alpha$-successors of $X_k$ is denoted by $S^\alpha (X_k)$.

Consider the signified graph $\spk{5}$ depicted in Figure~\ref{fig:SP5}.
For example, given $\alpha=(+1,-1)$ and $X=(0,3)$, the vertex $1$ is an $\alpha$-successor
of $X$, the vertex $2$ is an $\alphab$-successor of $X$, we have
$S^\alpha(X)=\set{1}$ and $S^{\alphab}(X)=\set{2}$. 

A signified graph $(G)$ has property $\Prop{k,l}$ 
if $|S^\alpha (X_k)|\ge l$ 
for any sequence $X_k$ of $k$ distinct vertices inducing a clique of $G$
and for any signed-vector $\alpha$ of size~$k$. 

\begin{lemma}\label{spq2tr}
 If $\spq$ has property $\Prop{n-1,k}$, then $\tr$ has property $\Prop{n,k}$.
\end{lemma}

\begin{proof}
  Suppose that $\spq$ has property $\Prop{n-1,k}$ and let
  $\alpha=(\alpha_1,\alpha_2,\ldots,\alpha_n)$ be a given signed
  vector.  Let $X=(u_1,u_2,\ldots,u_{n-1},w)$ be $n$ distinct vertices
  inducing a clique of $\tr$.  We have to prove that $X$ admits $k$
  $\alpha$-successors.  By noticing that
  $S^{\alphab}(X)=\twin(S^{\alpha}(X))$, we restrict the proof to the
  case $\alpha_n=+1$.  We define $X'=(v_1,v_2,\ldots,v_{n-1},w)$ such
  that $v_i=u_i$ if $u_iw$ is a positive edge and $v_i=\twin(u_i)$ if
  $u_iw$ is a negative edge.  Hence, $X'$ is a set of $n$ distinct
  vertices of $\tr$ such that $\bigcup_i v_i\subseteq N^+(w)$. By
  Lemma~\ref{lem:transitivity}, $\tr$ is vertex-transitive and thus
  $N^+(w)\cong K_q\cong \spq$.  Therefore the $(n-1)$ vertices
  $X''=X'\setminus\set{w} = (v_1,v_2,\ldots,v_{n-1})$ form a subset of
  some $V(\spq)$. Then by Property $\Prop{n-1,k}$ of $\spq$, there
  exist $k$ $(\alpha'_1,\alpha'_2,\ldots,\alpha'_{n-1})$-successors
  $x_1,x_2,\ldots, x_k$ of $X''$ in $\spq$, with $\alpha'_i=\alpha_i$
  (resp. $\alpha'_i=-\alpha_i$) if $v_i=u_i$ (resp. if
  $v_i=\twin(u_i)$). The $x_i$'s are clearly positive neighbors of $w$
  and hence, they are
  $(\alpha'_1,\alpha'_2,\ldots,\alpha'_{n-1},\alpha_{n})$-successors
  of $X'$.  So $X$ has $k$ $\alpha$-successors.
\end{proof}

\begin{lemma}\label{lem:double-tromp}
 If $(G,\Sigma)$ is a signified graph and $Tr(G,\Sigma)$ has property
 $\Prop{n,k}$, then $AT(G,\Sigma)$ has property $\Prop{n,k-1}$. 
\end{lemma}

\begin{proof}
 Recall that $Tr(G,\Sigma)$ is built from two isomorphic copies of
 $(G,\Sigma)$ plus two vertices $\infty_0$ and $\infty_1$. The graph
 $AT(G,\Sigma)$ is obtained from $Tr(G,\Sigma)$ by removing both
 $\infty_0$ and $\infty_1$. Now suppose $Tr(G,\Sigma)$ has property
 $\Prop{n,k}$. This means that any $n$ distinct vertices inducing a
 clique in $Tr(G,\Sigma)$ has $k$ $\alpha$-successors for any
 signed $n$-vector $\alpha$. Let $X\subseteq V(Tr(G,\Sigma))$ be
 any sequence of $n$ distinct vertices inducing a clique in
 $\tr$ such that both $\infty_0$ and $\infty_1$ do not belong to $X$.
 Then, for any signed $n$-vector $\alpha$, the set of $k$
 $\alpha$-successors $S^\alpha(X)$ cannot contains both $\infty_0$ and
 $\infty_1$. Then, it is clear that $X$ has at least $k-1$
 $\alpha$-successors in $AT(G,\Sigma)$.
\end{proof}

\begin{lemma}{\ }
\label{lem:ppt}
 \begin{enumerate}
 \item $\spq$ has properties $\Prop{1,\frac{q-1}{2}}$ and $\Prop{2,\frac{q-5}{4}}$.
 \item $\tr$ has properties $\Prop{1,q}$,
  $\Prop{2,\frac{q-1}{2}}$, and $\Prop{3,\frac{q-5}{4}}$.
 \item $AT(\spq)$ has properties $\Prop{1,q-1}$,
  $\Prop{2,\frac{q-3}{2}}$, and
  $\Prop{3,\max\paren{0,\frac{q-9}{4}}}$.
 \end{enumerate}
\end{lemma}

\begin{proof}
 \begin{enumerate}
 \item These properties follow from the fact that the signified Paley
  graph $\spq$ is built from the Paley graph $P_q$ which is
  self-complementary, edge transitive and strongly regular with
  parameters $(q,\frac{q-1}{2},\frac{q-5}{4},\frac{q-1}{4})$.
 \item $\tr$ has property $\Prop{1,q}$ since it is vertex transitive
  by Lemma~\ref{lem:transitivity} and the vertex $\infty$ has $q$
  positive and $q$ negative neighbors. The other properties follow
  from (1) and Lemma~\ref{spq2tr}.
 \item These properties follow from (2) and Lemma~\ref{lem:double-tromp}.
 \end{enumerate}
\end{proof}

\section{Results on signified homomorphisms}\label{sec:signified}

This section is devoted to study signified homomorphisms of planar
graphs and outerplanar graphs.

An \emph{acyclic $k$-coloring} is a proper vertex-coloring
such that each cycle has at least three colors. In other words, the
graph induced by any two color classes is a forest.

In 1998, Alon and Marshall~\cite{am98} proved the following (the
signified graph $ZS_k$ has $k\cdot 2^{k-1}$ vertices and has been
considered in Section~\ref{subsec:ask}):

\begin{theorem}[\cite{am98}]\label{thm:alon-marshall}
  Let $(G,\Sigma)$ be such that $G$ admits an acyclic $k$-coloring.
  Then $(G,\Sigma)$ admits a signified homomorphism to $ZS_k$ and thus
  $\chi_2(G)\le {k\cdot 2^{k-1}}$.
\end{theorem}

Note that Theorem~\ref{thm:alon-marshall} is actually tight as shown
by Huemer et al.~\cite{hfmm08} in 2008.

The \emph{girth} of a graph is the length of a shortest cycle. We
denote by $\mathcal{P}_g$ (resp. $\mathcal{O}_g$) the class of planar
graphs (resp. outerplanar graphs) with girth at
least $g$ (note that $\mathcal{P}_3$ is simply the class of planar
graphs). 

Borodin~\cite{b79} proved that every planar graph admits an acyclic
$5$-coloring. We thus get the following from
Theorem~\ref{thm:alon-marshall}:
\begin{corollary}\label{cor:main3-montejano}
  Every planar graph admits a signified homomorphism to $ZS_5$. We
  thus have $\chi_2(\mathcal{P}_3)\le 80$.
\end{corollary}

In this same context, Montejano et al.~\cite{mopr10} obtained in 2010
the following results:

\begin{theorem}[\cite{mopr10}]\
  \begin{enumerate}\label{thm:main-montejano}
  \item Every planar graph of girth at least $5$ admits a signified
  homomorphism to $\trk{9}$. We thus have $\chi_2(\mathcal{P}_5)\le
  20$. \label{thm:main5-montejano}
\item Every planar graph of girth at least $6$ admits a signified
  homomorphism to $\trk{5}$. We thus have $\chi_2(\mathcal{P}_6)\le
  12$.\label{thm:main6-montejano}
  \end{enumerate}
\end{theorem}

In this section, we get the following new results. We obtain a first
result on outerplanar graph with girth at least $4$ (see
Theorem~\ref{thm:outer-signified}) and a second one on planar graph
with girth at least $4$ (see Theorem~\ref{thm:main-signified}). The
latter result gives a new upper bound on the signified chromatic
number. We then construct a planar graph with signified chromatic
number $20$ (see Theorem~\ref{thm:lower-planar}). We finally give
properties that must verify the target graphs for outerplanar and
planar graphs.

\begin{theorem}\label{thm:outer-signified}
  Every outerplanar graph with girth at least $4$ admits a signified
  homomorphism to $AT(K_4^*)$, where $K_4^*$ denotes the complete graph
  on $4$ vertices with exactly one negative edge.
\end{theorem}

\begin{proof}
  Assume by contradiction that there exists a counterexample to the
  result and let $(H,\Lambda)$ be a minimal counterexample in term of
  number of vertices. 

  Suppose $(H)$ contains a vertex $u$ of degree at most $1$. By
  minimality of $(H)$, the graph $(H')=(H\setminus\set{u})$ admits
  a signified homomorphism to $AT(K_4^*)$. Since every vertex of   
  $AT(K_4^*)$ is incident to a positive and a negative edge, we can
  extend the signified homomorphism to $(H)$, a contradiction. 
  
  Suppose that $(H)$ contains two adjacent vertices $u$ and $v$ of
  degree $2$. By minimality of $(H)$, the graph
  $(H')=(H\setminus\set{u,v})$ admits a signified homomorphism to
  $AT(K_4^*)$. One can check that for every pair of (non necessarily
  distinct) vertices $x$ and $y$ of $AT(K_4^*)$, there exist the $8$
  possible signified $3$-paths.  We can therefore extend the signified
  homomorphism to $(H)$, a contradiction.
  
  Pinlou and Sopena~\cite{ps06c} showed that every
  outerplanar graph with girth at least $k$ and minimum degree at
  least $2$ contains a face of length $l \geq k$ with at least $(l-2)$
  consecutive vertices of degree $2$. Therefore, the counterexample
  $(H)$ is not an outerplanar graph of girth $4$, a contradiction,
  that completes the proof.
\end{proof}

\begin{theorem}\label{thm:main-signified}
  Every planar graph with girth at least $4$ admits a signified
  homomorphism to $AT(\spk{25})$. We thus have $\chi_2(\mathcal{P}_4)\le 50$.
\end{theorem}

Let $n_3(G)$ be the number of $^\ge 3$-vertices in the graph
$G$.  Let us define the partial order $\preceq$.  Given two 
graphs $G_1$ and $G_2$, we have $G_1\prec G_2$ if and only
if one of the following conditions holds:
\begin{itemize}
 \item $n_3(G_1)<n_3(G_2)$.
 \item $n_3(G_1)=n_3(G_2)$ and $|V(G_1)|+|E(G_1)|<|V(G_2)|+|E(G_2)|$.
\end{itemize}

Note that the partial order $\preceq$ is well-defined and is a partial
linear extension of the minor poset.

Let $(H)$ be a signified graph that
does not admit a homomorphism to the signified graph $AT(\spk{25})$ and
such that its underlying graph $H$ is a triangle-free
planar graph which is minimal with respect to $\preceq$. In the
following, $H$ is given with its embedding in the plane. A weak
$7$-vertex $u$ in $H$ is a $7$-vertex adjacent to four $2$-vertices
$v_1,\cdots, v_4$ and three $^\ge 3$-vertices $w_1, w_2, w_3$ such
that $v_1$, $w_1$, $v_2$, $w_2$, $v_3$, $w_3$, and $v_4$ are
consecutive.

\begin{lemma}\label{con4}
 The graph $H$ does not contain the following configurations:
 \begin{enumerate}[(C1)]
 \item a $^\le 1$-vertex;\label{item:1}
 \item a $k$-vertex adjacent to $k$ $2$-vertices for $2\le k\le 49$;\label{item:2}
 \item a $k$-vertex adjacent to $(k-1)$ $2$-vertices for $2\le k\le 24$;\label{item:3}
 \item a $k$-vertex adjacent to $(k-2)$ $2$-vertices for $3\le k\le 12$;\label{item:4}
 \item a $3$-vertex; \label{item:5b}
 \item a $k$-vertex adjacent to $(k-3)$ $2$-vertices for $4\le k\le 6$;\label{item:5}
 \item two vertices $u$ and $v$ linked by two distinct $2$-paths,
  both paths having a $2$-vertex as internal vertex;\label{item:7}
 \item a $4$-face $wxyz$ such that $x$ is 2-vertex, $w$ and $y$ are
  weak $7$-vertices, and $z$ is a $k$-vertex adjacent to $(k-4)$
  2-vertices for $4\le k\le 9$;\label{item:9}
 \end{enumerate}
\end{lemma}

\begin{proof}
  The drawing conventions for a configuration C$k$ contained in a
  signed graph $(H)$ are the following. First note that, in
  Figures~\ref{fig:config1} and~\ref{fig:config2}, we only draw the
  underlying graph $H$ of $(H)$, i.e. we do not distinguish positive
  and negative edges. The neighbors of a white vertex in $H$ are
  exactly its neighbors in C$k$, whereas a black vertex may have other
  neighbors in $H$. Two or more black vertices in C$k$ may coincide in
  a single vertex in $H$, provided they do not share a common white
  neighbor.  Configurations C2 - C8 are depicted in
  Figures~\ref{fig:config1} and~\ref{fig:config2}.
  
  For each configuration, we suppose that $H$ contains the
  configuration and we consider a signified triangle-free graph $(H')$
  such that $H'\prec H$.  We only argue that $H'\prec H$ for
  configuration C5.  For every other configuration, we have that $H'$
  is a minor of $H$ and thus $H'\prec H$.  Therefore, by minimality of
  $(H)$, $(H')$ admits a signified homomorphism $f$ to $AT(\spk{25})$.
  Then we modify and extend $f$ to obtain a signified homomorphism of
  $(H)$ to $AT(\spk{25})$, contradicting the fact that $(H)$ is a
  counterexample.
  
  By Lemma~\ref{lem:ppt}, $AT(\spk{25})$ satisfies properties
  $\Prop{1,24}$, $\Prop{2,11}$, and $\Prop{3,4}$.

\begin{figure}
 \subfigure[\label{subfig:a}C\ref{item:2}]{\includegraphics[scale=1.2]{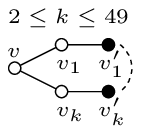}} \hfill
 \subfigure[\label{subfig:b}C\ref{item:3}]{\includegraphics[scale=1.2]{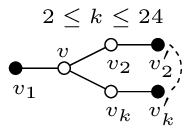}} \hfill
 \subfigure[\label{subfig:c}C\ref{item:4}]{\includegraphics[scale=1.2]{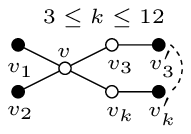}} \hfill
 \caption{Configurations C\ref{item:2}--C\ref{item:4}.\label{fig:config1}}
\end{figure}

\textbf{\emph{Proof of configuration C1:}} Trivial.

\textbf{\emph{Proof of configuration C2:}} Suppose that $(H)$ contains
the configuration depicted in Figure~\ref{subfig:a} and $f$ is a signified
homomorphism of $(H')=(H)\setminus\set{v,v_1,\cdots,v_k}$ to $AT(\spk{25})$.
For every $i$, if the edges $vv_i$ and $v_iv'_i$ have the same sign
(resp. different signs), then $v$ must get a color distinct from
$\twin(f(v'_i))$ (resp. $f(v'_i)$). So, each $v'_i$ forbids at most
one color for $v$. Thus there remains an available color for
$v$. Then we extend $f$ to the vertices $v_i$ using property $\Prop{2,11}$.

\textbf{\emph{Proof of configuration C3:}} Suppose that $(H)$ contains
the configuration depicted in Figure~\ref{subfig:b} and $f$ is a signified
homomorphism of $(H')=(H)\setminus\set{v,v_2,\cdots,v_k}$ to $AT(\spk{25})$.
As shown in the proof of Configuration C2, each $v'_i$ forbids at most
one color for $v$. So, we have at most $23$ forbidden colors for $v$
and by property $\Prop{1,24}$, there remains at least one available color for $v$. 
Then we extend $f$ to the vertices $v_i$ ($2\le i\le k$) using property $\Prop{2,11}$.

\textbf{\emph{Proof of configuration C4:}} Suppose that $(H)$ contains
the configuration depicted in Figure~\ref{subfig:c} and $f$ is a signified
homomorphism of $(H')=(H)\setminus\set{v_3,\cdots,v_k}$ to $AT(\spk{25})$.
As shown in the proof of Configuration C2, each $v'_i$ forbids at most one color for $v$.
So, we have at most $10$ forbidden colors for $v$ and by property
$\Prop{2,11}$, there remains at least one available color in order to recolor $v$.
Then we extend $f$ to the vertices $v_i$ ($3\le i\le k$) using property
$\Prop{2,11}$.

\textbf{\emph{Proof of configuration C5:}} Suppose that $(H)$ contains
the configuration depicted in Figure~\ref{subfig:db}.  Let $(H')$ be
the graph obtained from $(H)$ by deleting the vertex $v$ and by
adding, for every $1\le i < j\le 3$, a new vertex $v_{ij}$ and the
edges $v_iv_{ij}$ and $v_{ij}v_j$. Each of the $6$ edges $v_iv_{ij}$
gets the sign $\alpha_i$ of the edge $v_iv$ in $(H)$.
As configuration C4 is forbidden in $H$,
we have $d_H(v_i)\ge 3$ for $i\in\set{1,2,3}$.  We have $H'\prec H$
since $n_3(H')<n_3(H)$. Clearly, $H'$ is triangle free.  Hence, there
exists a signified homomorphism $f$ of $(H')$ to $AT(\spk{25})$.
By $\Prop{3,4}$, we can find an $\alpha$-successor $u$ of $(f(v_1),
f(v_2), f(v_3))$ in $AT(\spk{25})$ with
$\alpha=(\alpha_1,\alpha_2,\alpha_3)$. Now fix $f(v)=u$.  Note that
$f$ restricted to $V(H)$ is a homomorphism of $(H)$ to $AT(\spk{25})$.

\textbf{\emph{Proof of configuration C6:}} Suppose that $(H)$ contains
the configuration depicted in Figure~\ref{subfig:d} and $f$ is a
signified homomorphism of $(H')=(H)\setminus\set{v_4,\ldots,v_k}$ to
$AT(\spk{25})$.  As shown in the proof of Configuration C2, each $v'_i$
forbids at most one color for $v$. So, we have at most $3$ forbidden
colors for $v$ and by property $\Prop{3,4}$, there remains at least
one available color for $v$.  Then we extend $f$ to the vertices $v_i$
($4\le i\le k$) using property $\Prop{2,11}$.

\begin{figure}
 \subfigure[\label{subfig:db}C\ref{item:5b}]{\includegraphics[scale=1.2]{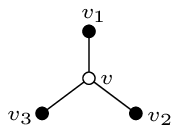}} \hfill
 \subfigure[\label{subfig:d}C\ref{item:5}]{\includegraphics[scale=1.2]{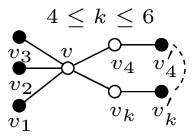}}\hfill
 \subfigure[\label{subfig:f}C\ref{item:7}]{\includegraphics[scale=1.2]{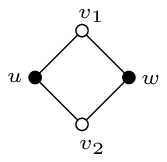}}\hfill
 \subfigure[\label{subfig:h}C\ref{item:9}]{\includegraphics[scale=1.2]{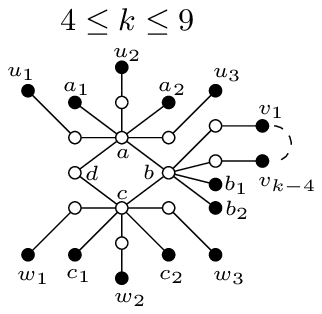}}
 \caption{Configurations C\ref{item:5b}--C\ref{item:9}.\label{fig:config2}}
\end{figure}

\textbf{\emph{Proof of configuration C7:}} 
Suppose that $(H)$ contains
the configuration depicted in Figure~\ref{subfig:f}.

If $u$ and $w$ have no common neighbor other than $v_1$ and $v_2$,
then we consider the graph $(H')$ obtained from
$(H)\setminus\set{v_1,v_2}$ by adding the edge $uw$.

If $u$ and $w$ have at least one other common neighbor $v_3$, then
consider the graph $(H')$ obtained from $(H)\setminus\set{v_1,v_2}$ by
adding a vertex $v$ adjacent to $u$ and $w$ such that $uv$ is negative
and the sign of $vw$ is the product of the signs of $uv_3$ and
$v_3w$. Therefore, we have at least two $2$-paths linking $u$ and $w$,
one whose both edges have the same sign and one whose edges have
different sign.

In both cases, $H'$ is triangle free and is a minor of $H$, so that
$(H')$ admits a signified homomorphism $f$ to $AT(\spk{25})$.  Also, in
both cases, $f(u)$ and $f(w)$ form an edge in $AT(\spk{25})$ since
$f(u)\neq f(v)$ and $f(u)\neq f(\twin(v))$.  Thus,
the coloring of $(H)\setminus\set{v_1,v_2}$ induced by $f$ can be
extended to $(H)$ using property $\Prop{2,11}$.

\textbf{\emph{Proof of configuration C8:}} Suppose that $(H)$ contains
the configuration depicted in Figure~\ref{subfig:h}.  By
Corollary~\ref{cor:bound_signed}(\ref{cor:bound_signed-2}), $(H)$
admits a signified homomorphism to $AT(\spk{25})$ if and only if every
equivalent signature of $(H)$ admits a signified homomorphism to
$AT(\spk{25})$.  So, by resigning a subset of vertices in
$\set{a,b,c,c_1,c_2}$, we can assume that the edges $da$, $ab$, $bc$,
$cc_1$, and $cc_2$ are positive.  Consider a signified homomorphism
$f$ of $(H')=(H\setminus\set{d})$ to $AT(\spk{25})$.  The edge $dc$ in
$(H)$ has to be negative, since otherwise $f$ would be extendable to
$(H)$ by setting $f(d)=f(b)$.  Also, we must have $f(c)=f(a)$, since
otherwise we could color $d$ using property $\Prop{2,11}$.  Now, we
show in the remaining of the proof that we can modify $f$ such that
$f(c)\ne f(a)$.  Let us define the signed vector $\alpha=(+1,+1,+1)$.
We assume that $f(b)$, $f(c_1)$ and $f(c_2)$ are distinct, since the
case when they are not distinct is easier to handle.

For $1\le i\le 3$, let $k_i$ denote the color that is forbidden for
$c$ by $w_i$, that is, $k_i=f(w_i)$ if the edges of the 2-path linking
$c$ and $w_i$ have distinct signs and $k_i=\twin(f(w_i))$ otherwise.
By property $\Prop{3,4}$, the sequence $X=(f(c_1),f(c_2),f(b))$ has at
least 4 $\alpha$-successors.  Assume that $X$ has at least 5
$\alpha$-successors. Then we can give to $c$ a color distinct from $k_1$,
$k_2$, $k_3$, and $f(a)$ that leads to $f(c)\ne f(a)$.

Assume now that $X$ has exactly $4$ $\alpha$-successors. The graph
$AT(\spk{25})$ contains two copies of $\spk{25}$, namely $\spk{25}^0$
and $\spk{25}^1$. Suppose that $X$ is not contained in one copy of
$\spk{25}$.  We consider the graph $Tr(\spk{25})$ obtained by adding
the anti-twin vertices $\infty_0$ and $\infty_1$ to $AT(\spk{25})$.  By
Lemma~\ref{lem:ppt}, $Tr(\spk{25})$ satisfies $\Prop{3,5}$, so $X$
admits at least 5 $\alpha$-successors in $Tr(\spk{25})$.  Since $X$ is
not contained in one copy of $\spk{25}$ of the subgraph $AT(\spk{25})$,
the extra vertices $\infty_0$ and $\infty_1$ are not
$\alpha$-successors of $X$.  This means that $X$ has at least 5
$\alpha$-successors in $AT(\spk{25})$, contradicting the hypothesis.
So, without loss of generality, $X$ is necessarily contained in $\spk{25}^0$.

We represent the field $\mathbb{F}_{25}$ by the numbers $a+b\sqrt{2}$,
where $a$ and $b$ are integers modulo 5.
Without loss of generality, we can assume that $f(c_1)=0$ and
$f(c_2)=1$ since $\spk{25}$ is edge-transitive. Table~\ref{tab:succ}
gives the sequences $X$ having exactly 4 $\alpha$-successors together
with their 4 $\alpha$-successors.  

\def\figurename{Tab.}
\begin{figure}
  \centering
    $$\begin{array}{|l|cccc|} \hline
      (0_0,1_0,2_0) & 3_0, & 4_0, & (1+2\sqrt{2})_1, & (1+3\sqrt{2})_1\\ \hline
      (0_0,1_0,3_0) & 2_0, & 4_0, & (3+\sqrt{2})_1, & (3+4\sqrt{2})_1\\ \hline
      (0_0,1_0,4_0) & 2_0, & 3_0, & (2\sqrt{2})_1, & (3\sqrt{2})_1\\ \hline
      (0_0,1_0,(3+2\sqrt{2})_0) & (3+\sqrt{2})_1, & (3\sqrt{2})_1, & (1+3\sqrt{2})_1, & (3+4\sqrt{2})_1\\ \hline
      (0_0,1_0,(3+3\sqrt{2})_0) & (3+\sqrt{2})_1, & (2\sqrt{2})_1, & (1+2\sqrt{2})_1, & (3+4\sqrt{2})_1\\ \hline
    \end{array}$$
    \caption{Sets of the form $(0_0,1_0,x_0)$ having exactly 4
      $(+1,+1,+1)$-successors in $AT(\spk{25})$.\label{tab:succ}}
\end{figure}
\def\figurename{Fig.}

We are now ready to modify $f$. We decolor the vertices $a$, $b$, and
$c$. By property $\Prop{2,11}$, there exist at least two colors
$\beta$ and $\beta'$ for $b$ that are distinct from the colors
forbidden by the $k$ vertices $v_1,\cdots,v_{k-4},a_1,a_2,c_1,c_2$.
By previous discussions, the sequences $(0,1,\beta)$ and
$(0,1,\beta')$ must have exactly 4 $\alpha$-successors, so we can
assume that
$\set{\beta,\beta'}\subset\set{2,3,4,3+2\sqrt{2},3+3\sqrt{2}}$.  Let
us set $f(b)=\beta$. By property $\Prop{3,4}$, we can color $a$ such
that $f(a)$ is distinct from the colors forbidden by $u_1$, $u_2$,
$u_3$.  Now, $f$ is not extendable to $c$ and $d$ only if the 4
$\alpha$-successors of $(0,1,\beta)$ are $k_1$, $k_2$, $k_3$, and
$f(a)$. In particular, $k_1$, $k_2$, and $k_3$ have to be
$\alpha$-successors of $(0,1,\beta)$. Similarly, we can set
$f(b)=\beta'$ and obtain that $k_1$, $k_2$, and $k_3$ have to be
$\alpha$-successors of $(0,1,\beta')$ as well. This is a
contradiction, since we can observe that no two distinct sequences in
Table~\ref{tab:succ} have three common $\alpha$-successors.
\end{proof}

\begin{proofof}{Theorem~\ref{thm:main-signified}}
  Let $(H)$ be a minimal counterexample which is minimal with respect
  to $\preceq$. By Lemma~\ref{con4}, $(H)$ does not contain
  Configurations $C1$ to $C8$. There remains to show that every
  triangle-free planar graph contains at least one these $8$
  configurations. This has been already done using a discharging
  procedure in the proof of Theorem~2 in~\cite{op13}, where slightly
  weaker configurations were used.
\end{proofof}

We now exhibit the following lower bound for the signified chromatic
number of planar graphs:

\begin{theorem}\label{thm:lower-planar} 
  There exist planar graphs with signified chromatic number $20$.
\end{theorem}

\begin{proof}
  Let $(G_1,\Sigma_1)$ be the $6$-path $abcdef$ with
  $\Sigma=\set{bc,de}$ and let $(G_2,\Sigma_2)$ be the $6$-path
  $abcdef$ with $\Sigma=\set{ab,cd,ef}$.  Note that
  $\chi_2(G_1,\Sigma_1) =$ $\chi_2(G_2,\Sigma_2)$ $= 4$.

  Let $(G_3)$ be the outerplanar graph obtained from $(G_1,\Sigma_1)$,
  $(G_2,\Sigma_2)$ and a vertex $u$ such that $u$ is positively (resp.
  negatively) linked to the six vertices of $(G_1)$ (resp. $(G_2)$).
  For any signified coloring of $(G_3)$, we need $4$ colors for the
  vertices of $(G_1)$, $4$ other colors for the vertices of $(G_2)$,
  and a ninth color for $u$. We therefore have $\chi_2(G_3) = 9$.

  \begin{figure}
    \begin{center}
      \includegraphics[scale=0.5]{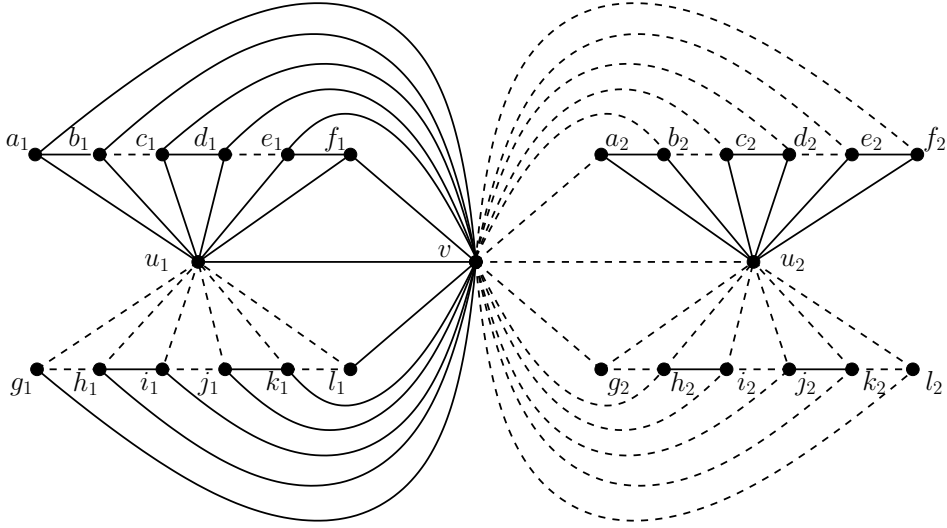}
      \caption{A planar graph with signified chromatic number 19\label{fig:lower}}
    \end{center}
  \end{figure}
  
  Let $(G_4)$ be the graph obtained from two copies of
  $(G_3)$ plus a vertex $v$ such that $v$ is positively linked to
  $13$ vertices of the first copy of $(G_3)$ and negatively linked
  to the $13$ vertices of the second one. This graph is depicted in
  Figure~\ref{fig:lower}. Once again, it is easy to check that
  $\chi_2(G_4) = 19$.
  
  Finally, let $(G_5)$ be the graph obtained from 28 copies
  $(G_4)_0,(G_4)_1,\ldots,(G_4)_{27}$ of $(G_4)$ as follows: we glue
  on each of the 27 vertices of $(G_4)_0$ the vertex $v$ of a copy
  $(G_4)_i$. Since $\chi_2(G_4) = 19$, we have $\chi_2(G_5) \ge 19$.
  Suppose, $\chi_2(G_5) = 19$. Therefore, there exists a graph $(H)$
  on $19$ vertices such that $(G_5)$ admits a signified homomorphism
  to $(H)$.  Moreover, since we glued a copy of $(G_4)$ on each
  vertices of $(G_4)_0$, then each color (i.e. each vertex of $(H)$)
  must have $9$ distinct positive neighbors and $9$ distinct negative
  neighbors. The subgraph of $(H)$ induced by the positive edges is a
  $9$-regular graph on $19$ vertices. Such a graph does not exist
  since, in every graph, the number of vertices of odd degree must be
  even.
\end{proof}

Montejano et al.~\cite{mopr10} proved that any outerplanar graph
admits a signified homomorphism to $SP_9$, that gives $\chi_2(G)\le 9$
whenever $G$ is an outerplanar graph. They also proved that this bound is
tight. We prove here that $SP_9$ is the only suitable target graph on
$9$ vertices.

\begin{theorem}\label{thm:outerplanar}
  The only graph of order $9$ to which every outerplanar graph admits
  a signified homomorphism is $SP_9$.
\end{theorem}

\begin{proof}
  Let $(G_1)$, $(G_2)$ and $(G_3)$ be the graphs constructed in the
  proof of Theorem~\ref{thm:lower-planar}.
  
  Note that since $\chi_2(G_1,\Sigma_1) = 4$, then for any signified
  $4$-coloring of $(G_1,\Sigma_1)$, among the three positive edges
  $ab$, $cd$ and $ef$, two of them will use $4$ distinct colors. We
  will later refer to this property in the remainder of this proof as
  Property 1. Using the same arguments as the previous
  paragraph, we have that, for any signified $4$-coloring of
  $(G_2,\Sigma_2)$, among the three negative edges $ab$, $cd$ and
  $ef$, two of them will use $4$ distinct colors. We will later refer
  to this property in the remainder of this proof as Property 2.
  
  Let $(G'_4)$ be the outerplanar graph obtained from 14 copies
  $(G_3)_0,(G_3)_1,\ldots,$ $(G_3)_{14}$ of $(G_3)$ as follows: we glue
  on each of the 13 vertices of $(G_3)_0$ the vertex $u$ of a copy
  $(G_3)_i$. Since $\chi_2(G) \le 9$ whenever $G$ is outerplanar, we
  have $\chi_2(G'_4) = 9$. Therefore, there exists a graph
  $(H_9)$ on $9$ vertices such that $(G'_4)$ admits a signified
  homomorphism to $(H_9)$. Each of the nine colors appears on the
  vertices of the copy $(G_3)_0$. Since we glued a copy of $(G_3)$
  on each vertex of $(G_3)_0$, then each color $c$ (i.e. each vertex
  $c$ of $(H_9)$) must have 4 positive neighbors
  $c^+_1,c^+_2,c^+_3,c^+_4$ and 4 negative neighbors
  $c^-_1,c^-_2,c^-_3,c^-_4$. Moreover, by Property~1 (resp.
  Property~2), we must have a positive (resp. negative) matching in
  the subgraph induced by $c^+_1,c^+_2,c^+_3,c^+_4$ (resp.
  $c^-_1,c^-_2,c^-_3,c^-_4$).
  
  Meringer~\cite{m99} provided an efficient algorithm to generate
  regular graphs with a given number of vertices and vertex degree. In
  particular, there exist $16$ connected $4$-regular graphs on $9$
  vertices (see Figure~\ref{fig:4reg}). Replacing edges by positive
  edges and non-edges by negative edges, these are the $16$ signified
  graphs such that each vertices have $4$ positive and $4$ negative
  neighbors. It is then easy to check that, among these $16$ graphs,
  there is only one graph with a positive (resp. negative) matching in
  the subgraph induced by the positive (resp. negative) neighbors of
  each vertex (see Figure~\ref{fig:4-reg-ok}): this is $SP_9$.
\end{proof}

\begin{figure}
  \subfigure[\label{fig:4reg1}]{\includegraphics[scale=0.3]{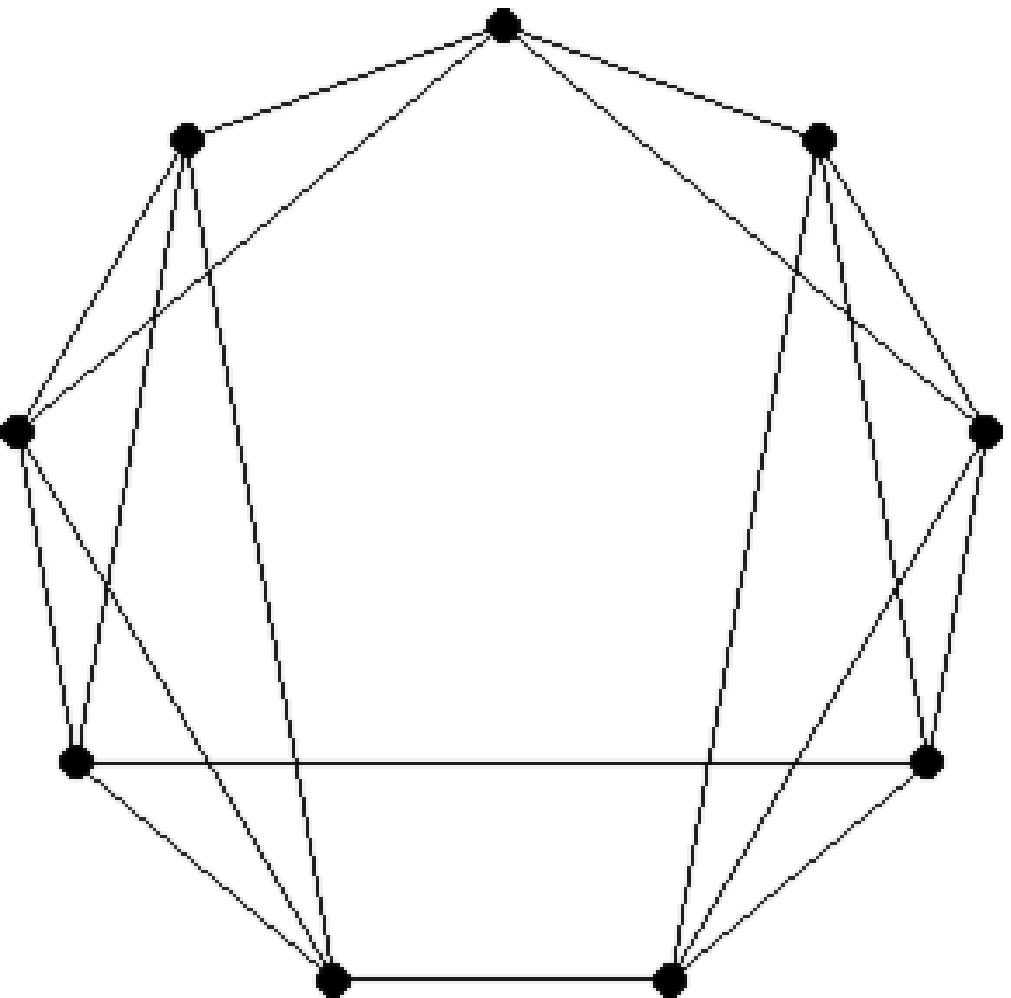}}\hfill
  \subfigure[\label{fig:4reg2}]{\includegraphics[scale=0.3]{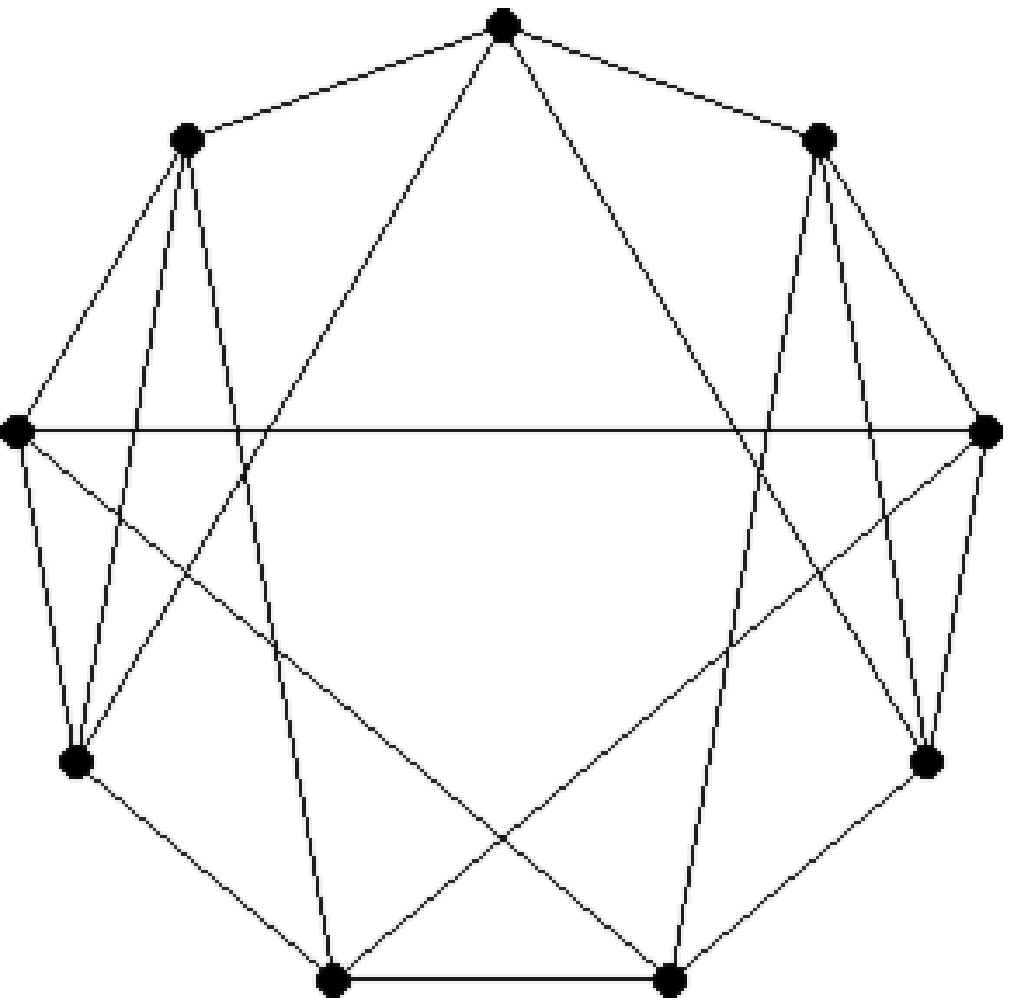}}\hfill
  \subfigure[\label{fig:4reg3}]{\includegraphics[scale=0.3]{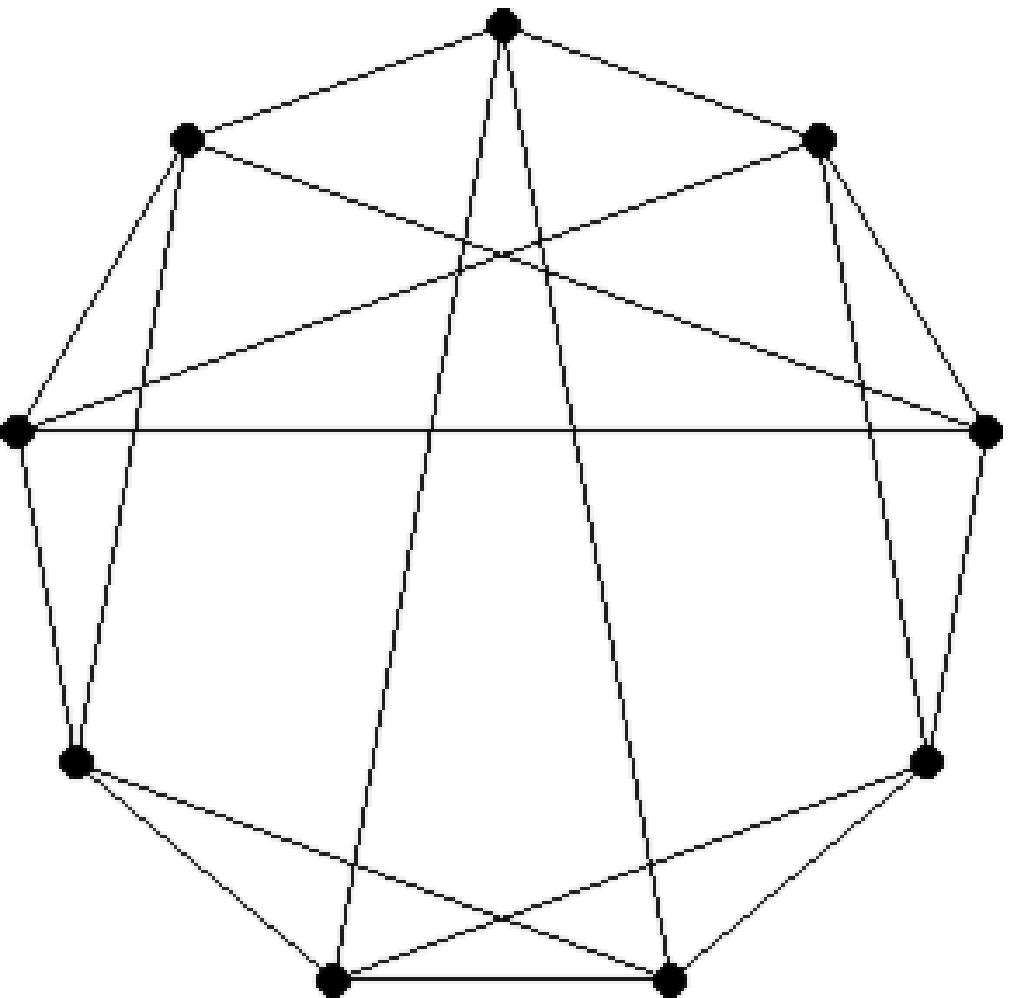}}\hfill
  \subfigure[\label{fig:4reg4}]{\includegraphics[scale=0.3]{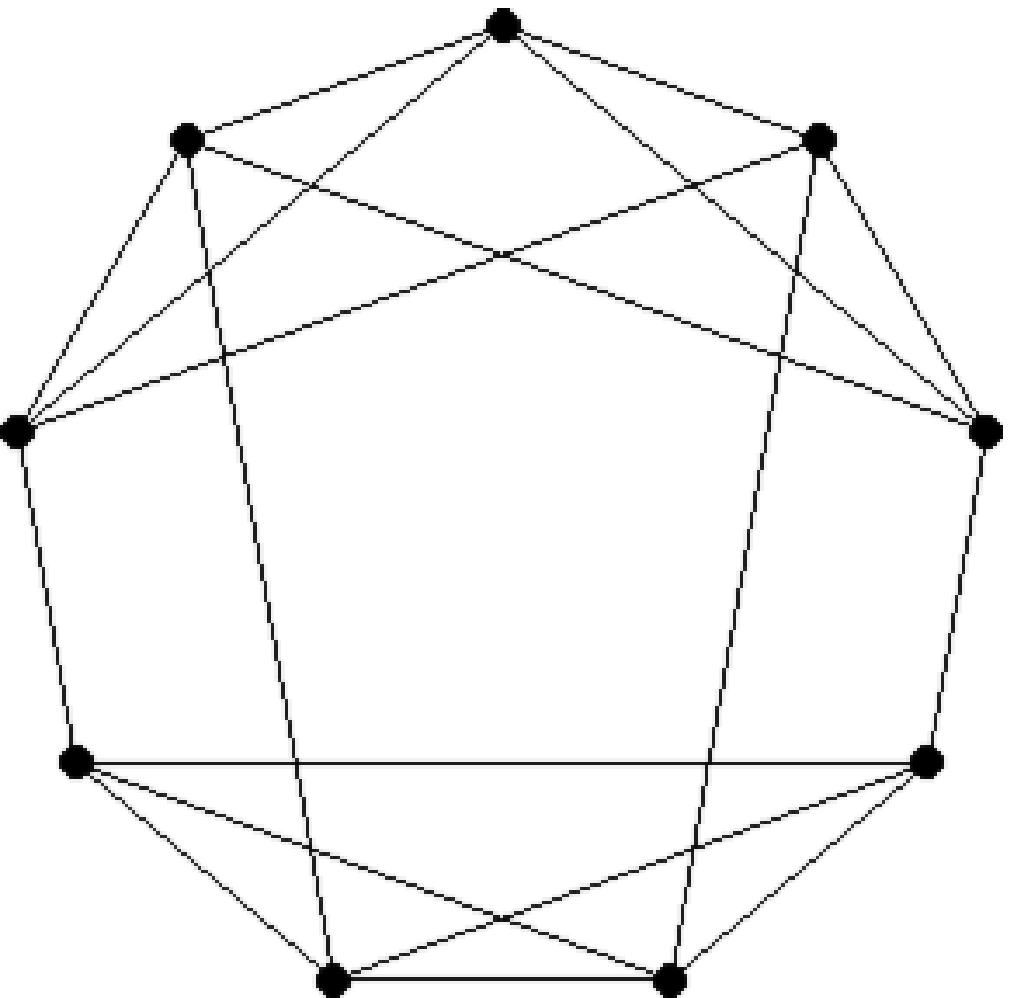}}

  \subfigure[\label{fig:4reg5}]{\includegraphics[scale=0.3]{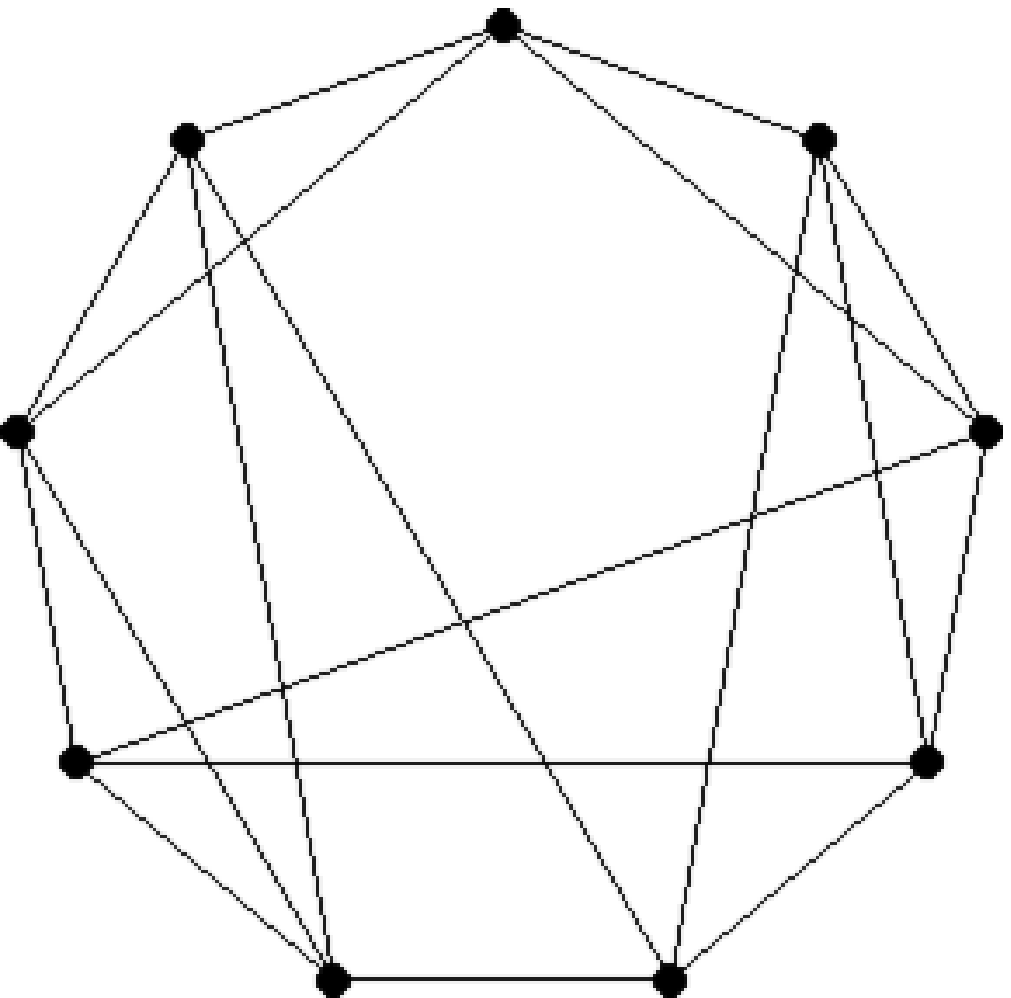}}\hfill
  \subfigure[\label{fig:4reg6}]{\includegraphics[scale=0.3]{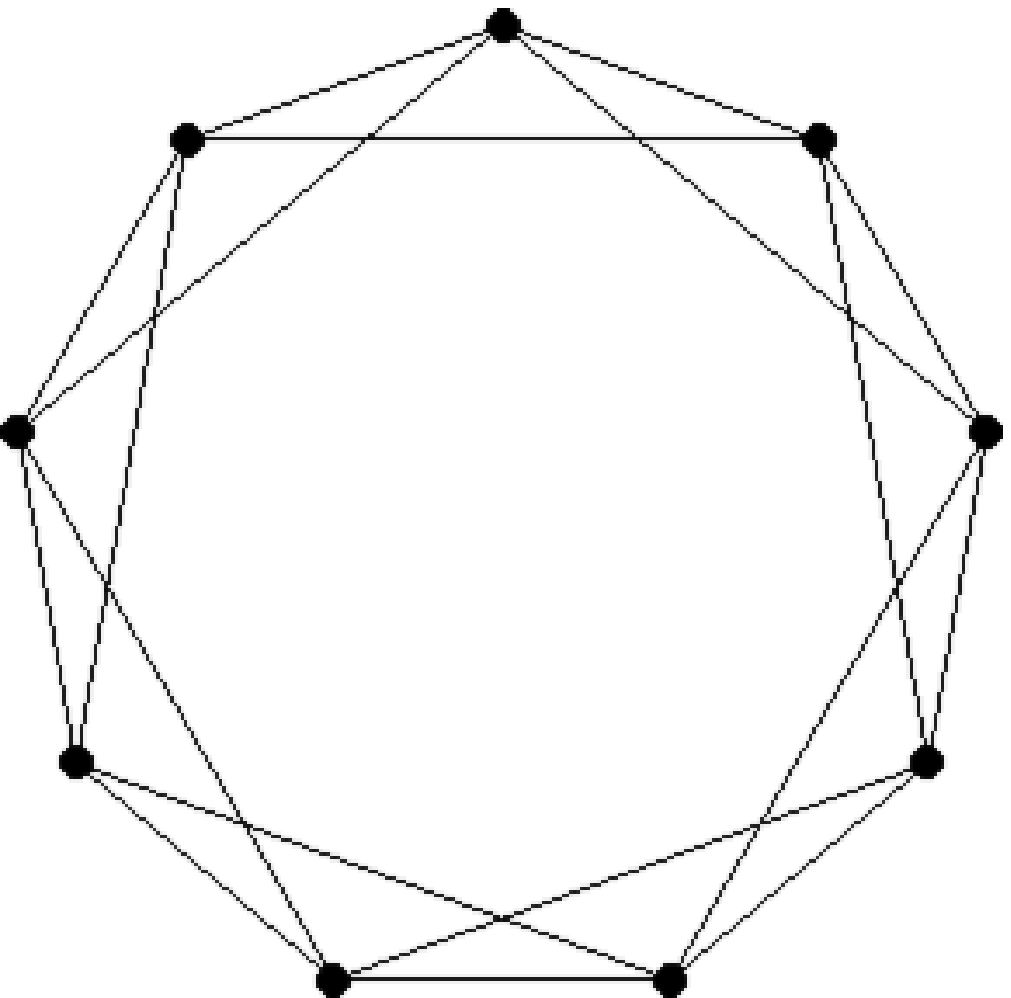}}\hfill
  \subfigure[\label{fig:4reg7}]{\includegraphics[scale=0.3]{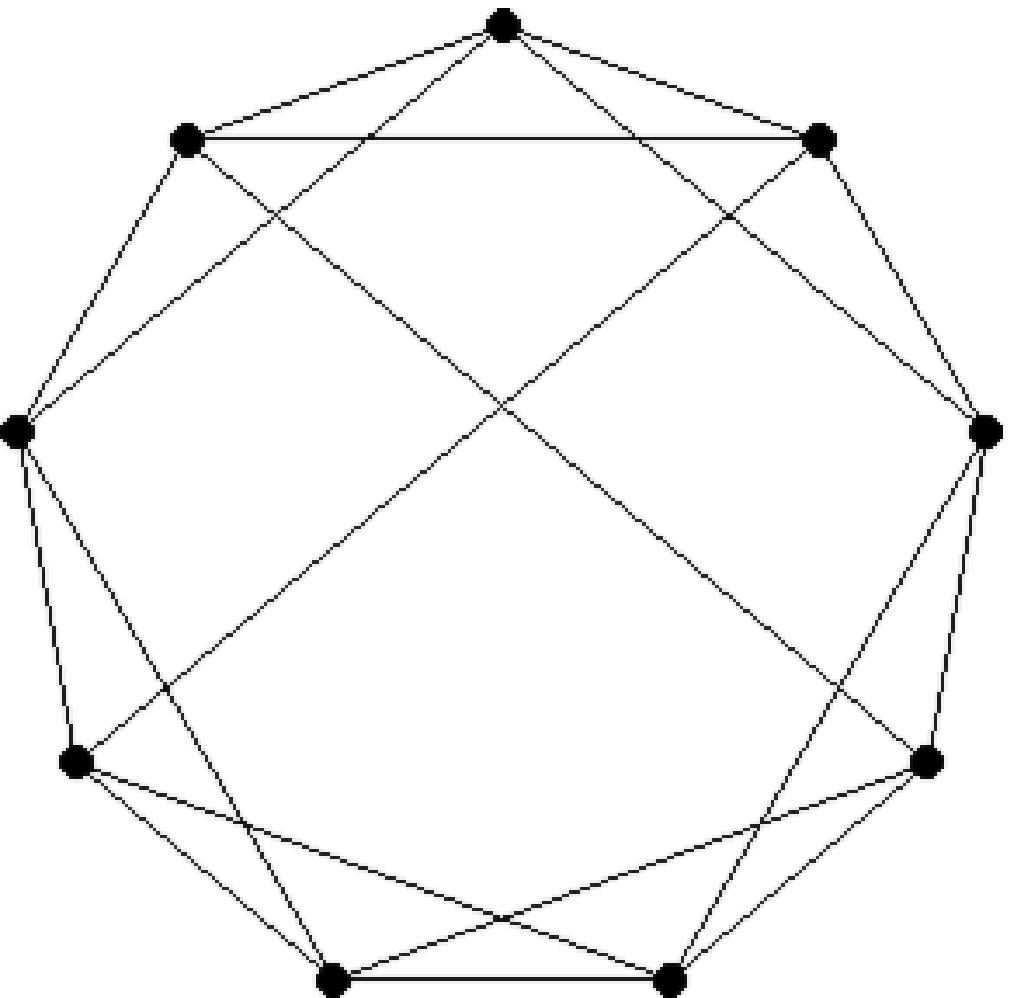}}\hfill
  \subfigure[\label{fig:4reg8}]{\includegraphics[scale=0.3]{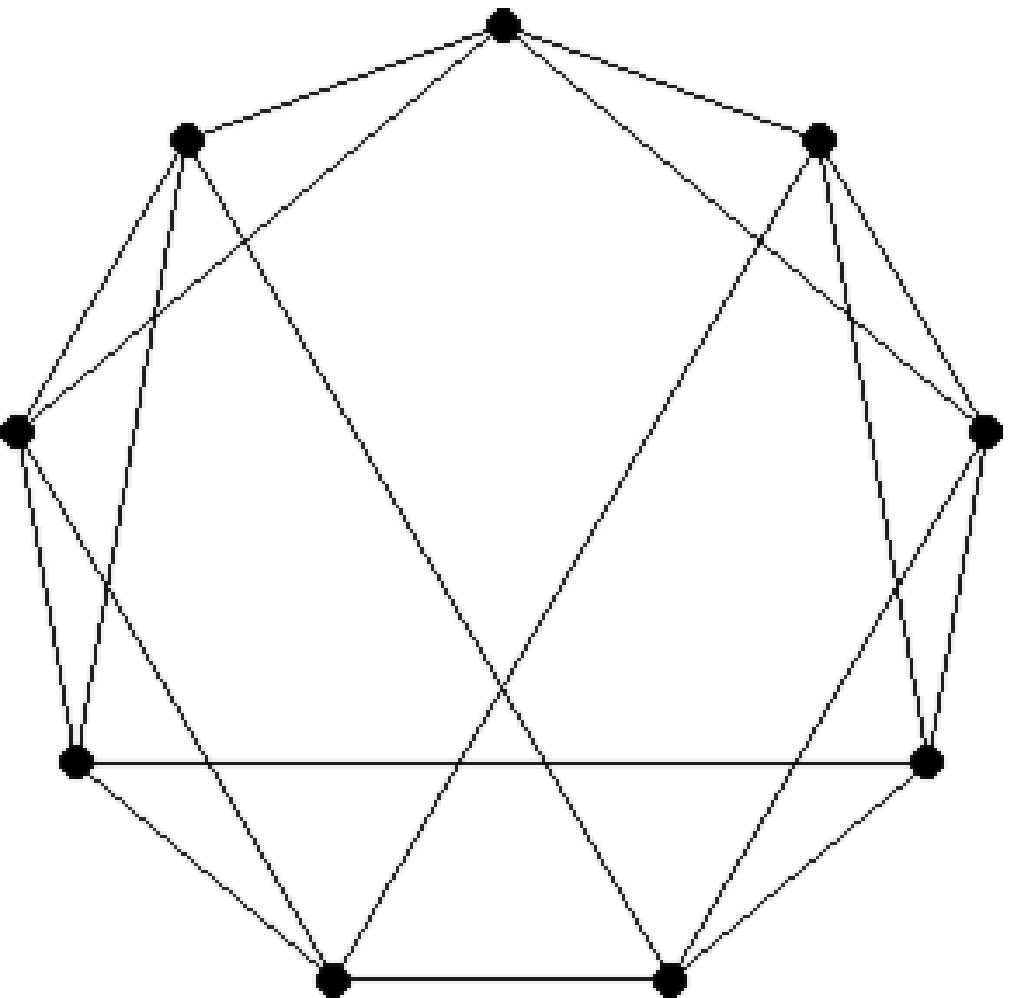}}

  \subfigure[\label{fig:4reg9}]{\includegraphics[scale=0.3]{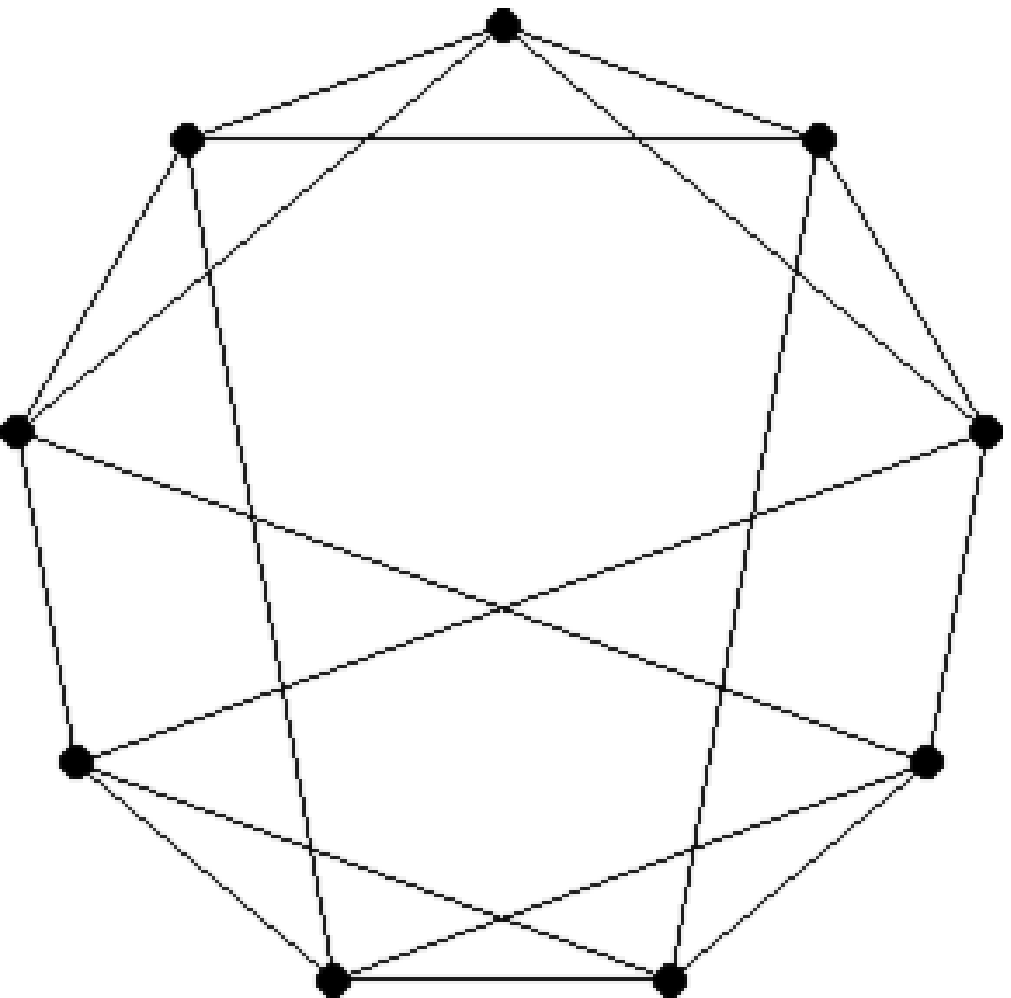}}\hfill
  \subfigure[\label{fig:4reg10}]{\includegraphics[scale=0.3]{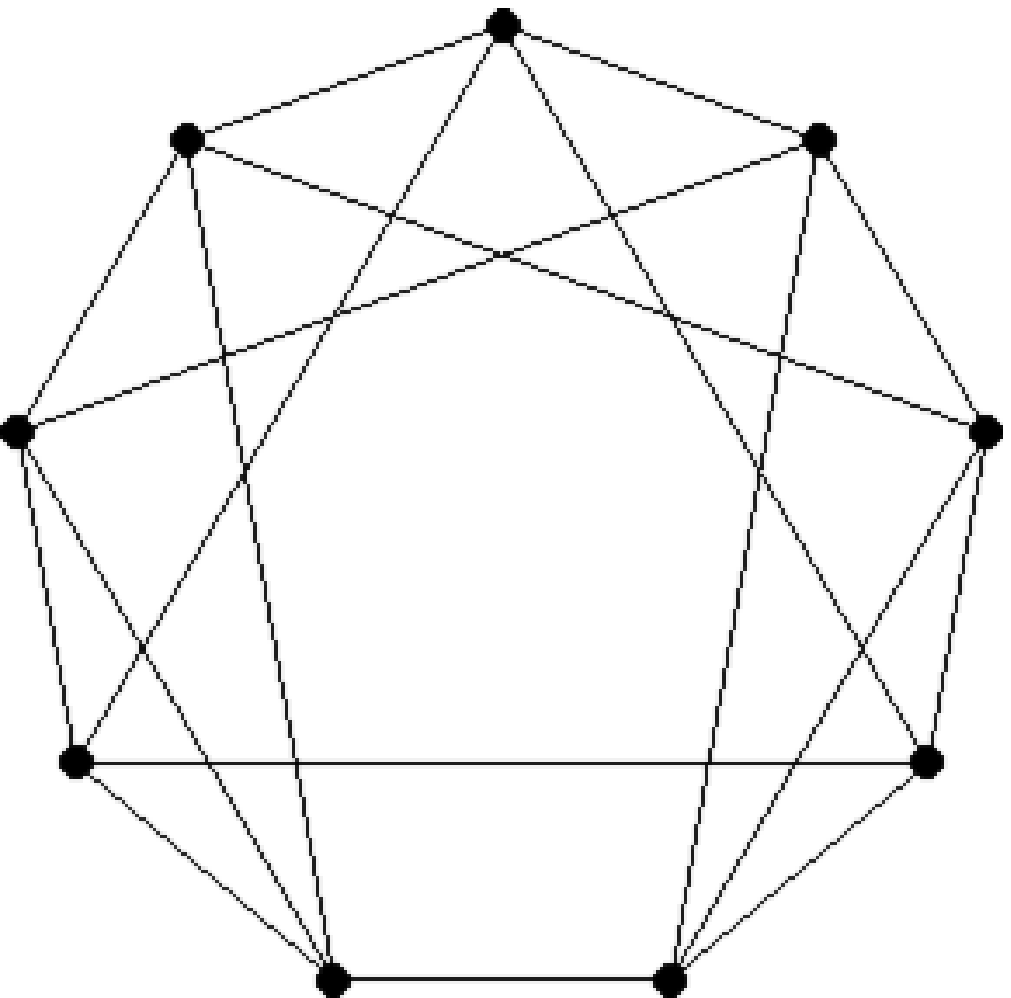}}\hfill
  \subfigure[\label{fig:4reg11}]{\includegraphics[scale=0.3]{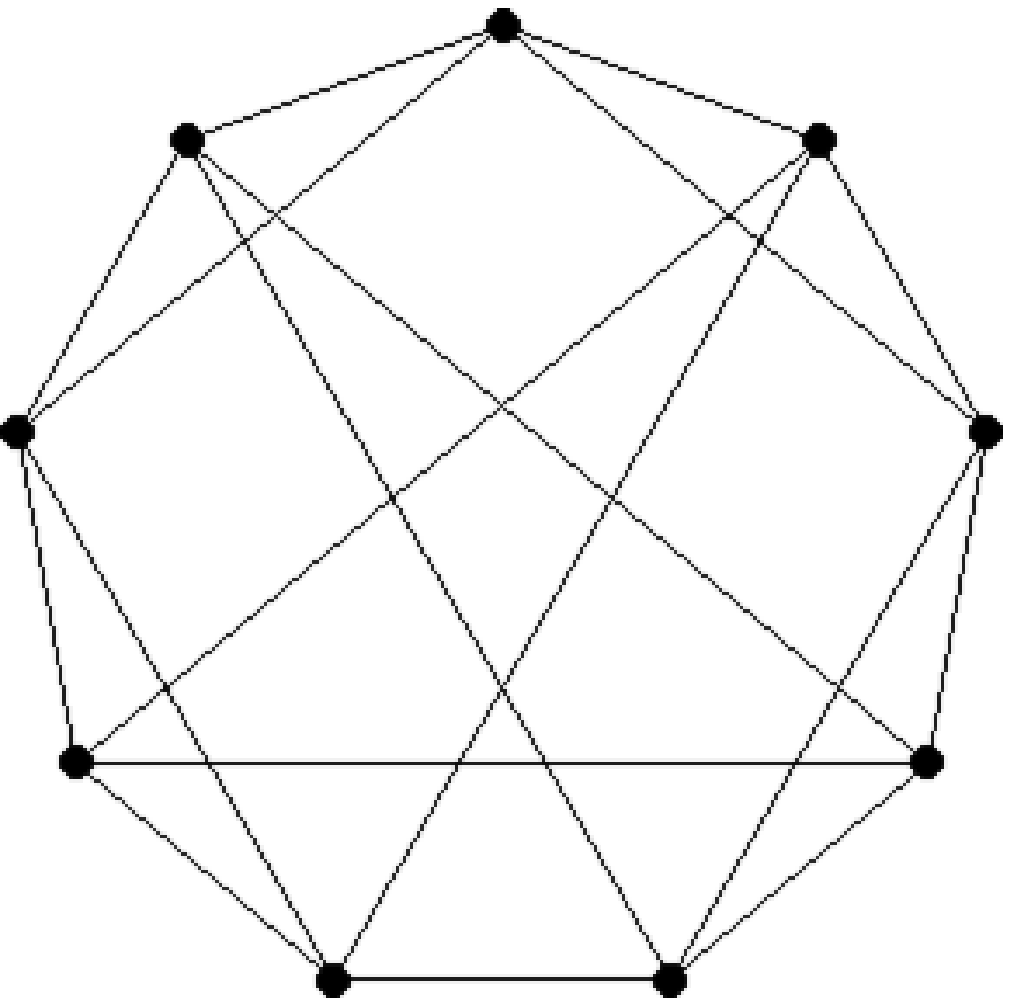}}\hfill
  \subfigure[\label{fig:4reg12}]{\includegraphics[scale=0.3]{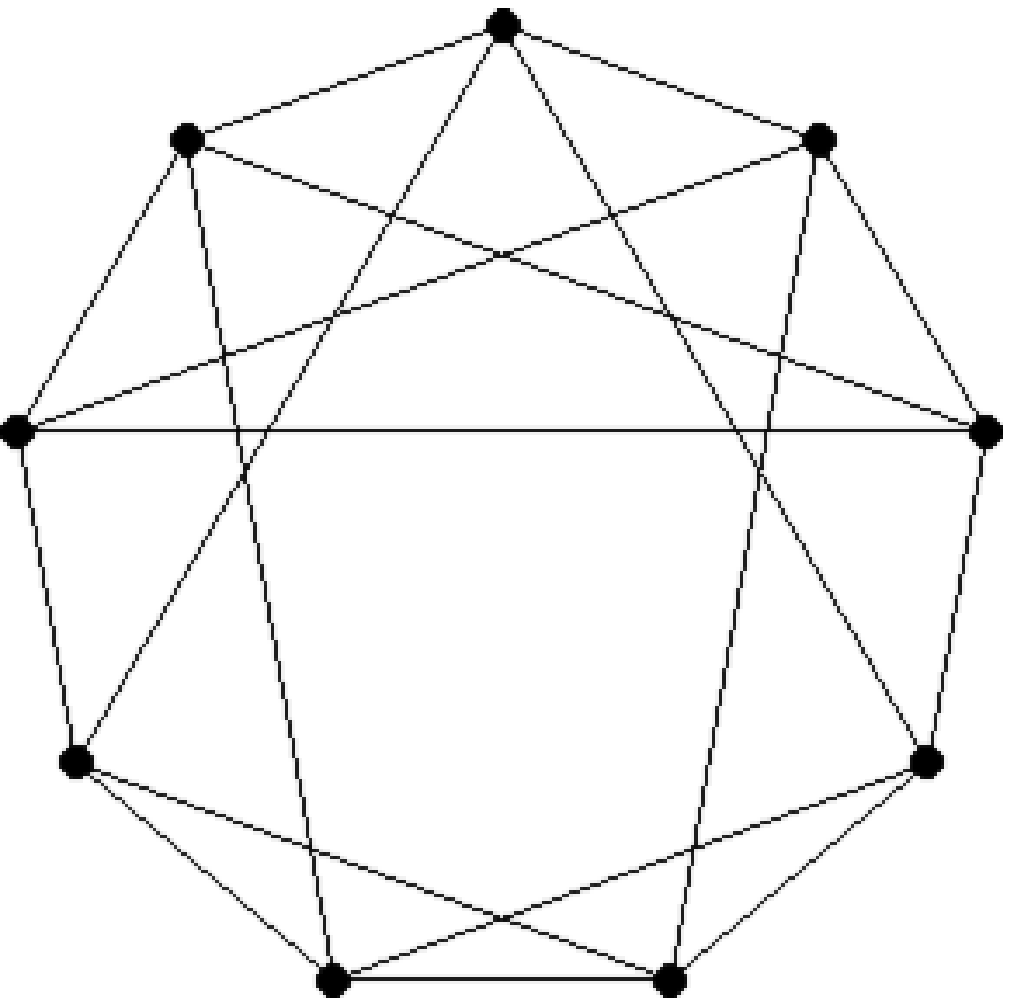}}

  \subfigure[\label{fig:4reg13}]{\includegraphics[scale=0.3]{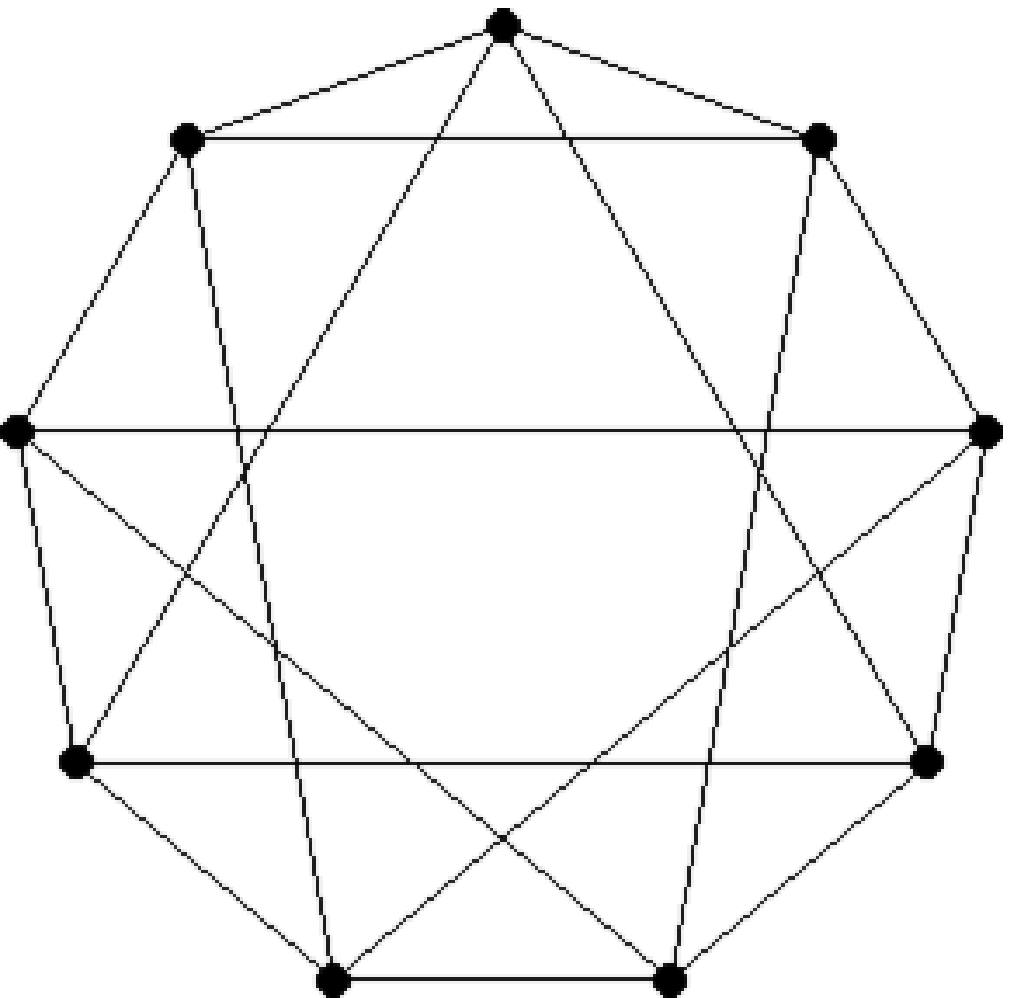}}\hfill
  \subfigure[\label{fig:4reg14}\label{fig:4-reg-ok}]{\includegraphics[scale=0.3]{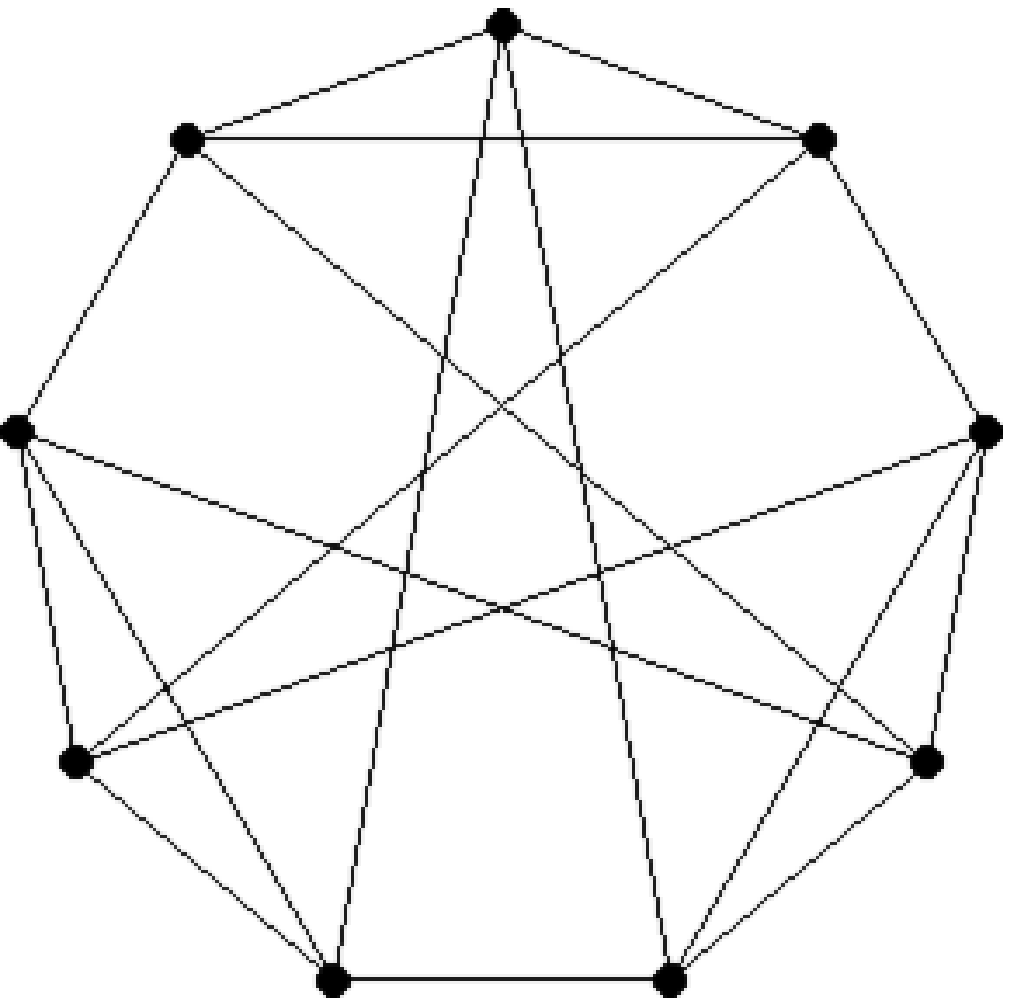}}\hfill
  \subfigure[\label{fig:4reg15}]{\includegraphics[scale=0.3]{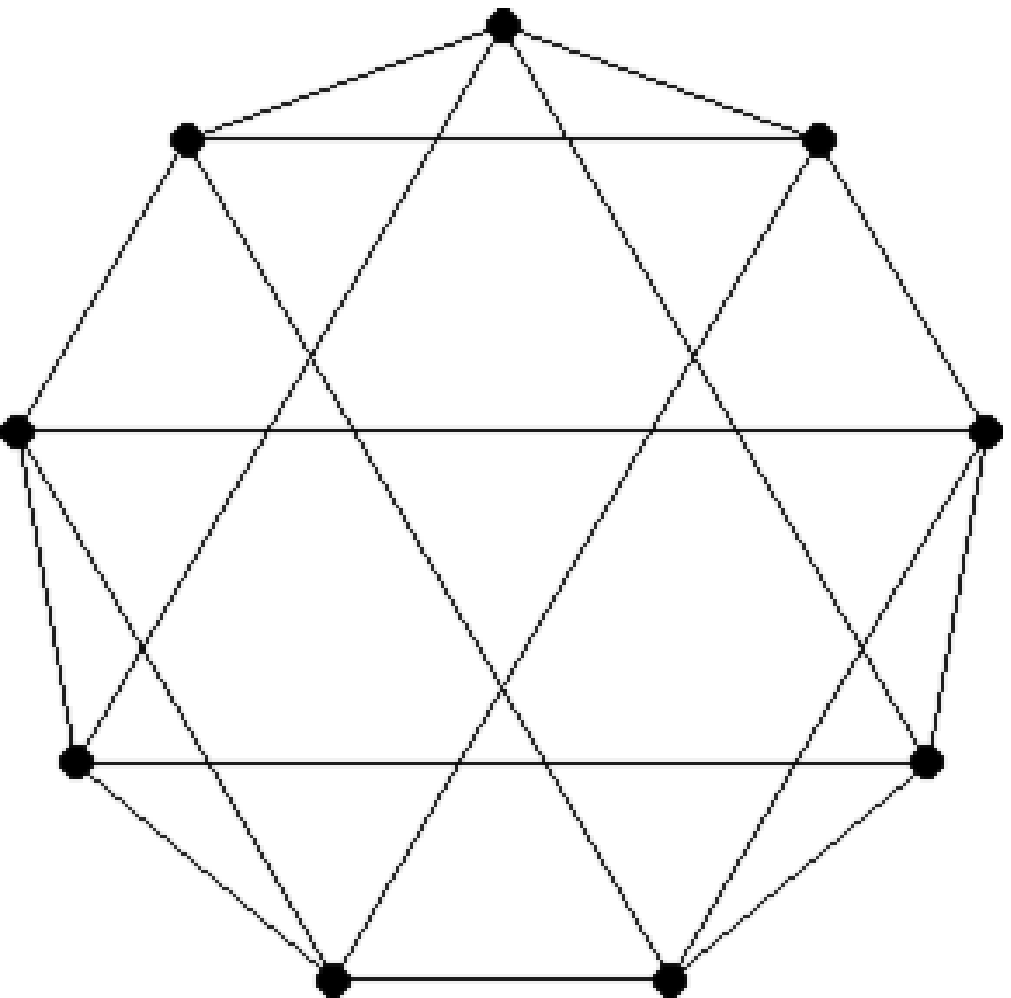}}\hfill
  \subfigure[\label{fig:4reg16}]{\includegraphics[scale=0.3]{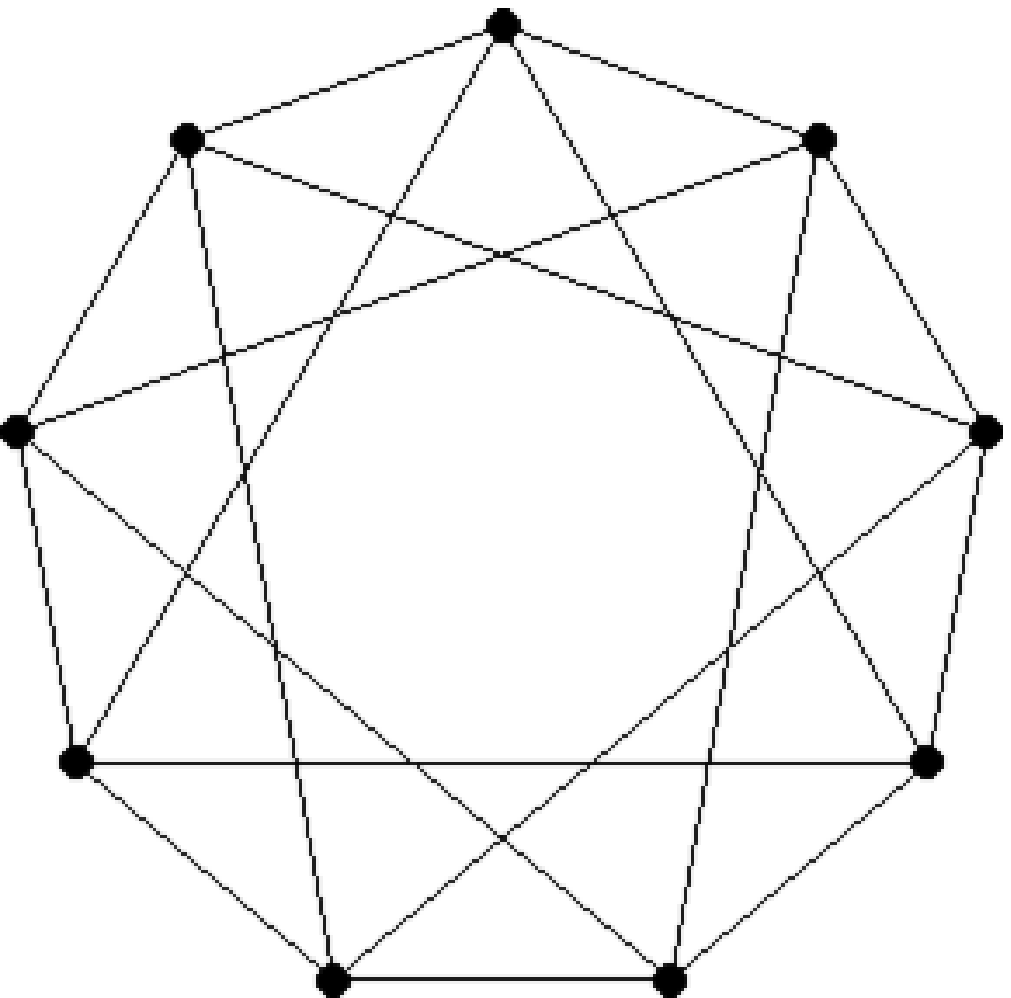}}

\caption{The $16$ connected $4$-regular graphs on $9$ vertices.\label{fig:4reg}}
\end{figure}

We finally give the same kind of result as
Theorem~\ref{thm:outerplanar} for planar graphs:

\begin{theorem}\label{thm:planar-20}
  If every planar graphs admits a signified homomorphism to an anti-twinned
  graph $H_{20}$ of order $20$, then $H_{20}$ is isomorphic to $\trk{9}$.
\end{theorem}

\begin{proof}
  Consider the outerplanar graph $(G'_4)$ of the proof of
  Theorem~\ref{thm:outerplanar}. Let $(G'_5)$ be the planar graph
  obtained from $(G'_4)$ plus a universal vertex positively linked
  to all the vertices of $(G'_4)$.
  
  The subgraph of $(G'_5)$ induced by the vertices of $(G'_4)$
  necessarily maps to the positive neighborhood of some vertex $v$ of
  $(H_{20})$. Since $(H_{20})$ is anti-twinned, every vertex has exactly
  $9$ positive neighbors. Therefore, by Theorem~\ref{thm:outerplanar},
  the positive neighborhood of $v$ is isomorphic to $SP_9$. Then the
  subgraph of $(H_{20})$ induced by $v$ and its positive neighborhood
  is $SP_9^+$. Since $SP_9^+$ is a clique of order $10$, it does not
  contain a pair of anti-twin vertices and thus $(H_{20})$ is isomorphic to
  $AT(SP_9^+)$. Then, by definition of Tromp signified Paley graphs,
  $(H_{20})$ is isomorphic to $\trk{9}$.
\end{proof}

\section{Results on signed homomorphisms}\label{sec:signed}

The upper bounds on signified chromatic number given in
Theorems~\ref{thm:alon-marshall},~\ref{thm:main-montejano},~\ref{thm:outer-signified}
and~\ref{thm:main-signified} are obtained by showing that the
considered graph class admits a signified homomorphism to some target
graph. As mentioned in
Corollary~\ref{cor:bound_signed}(\ref{cor:bound_signed-1}), when the
target graph of a signified homomorphism is an anti-twinned graph, this
gives a bound on the signed chromatic number. The upper bounds on
signed chromatic number given in this section are a direct consequence
of theses above-mentioned results.

Naserasr et al.~\cite{nrs12} proved the following:

\begin{theorem}[\cite{nrs12}]\label{thm:nrs-planar}
Let $G$ be a graph that admits an acyclic $k$-coloring. We have
$\chi_s(G)\le\ceil{\frac{k}{2}}\cdot 2^{k-1}$.
\end{theorem}

By Theorem~\ref{thm:alon-marshall}, $(G,\Sigma)$ admits a signified
homomorphism to $ZS_k$ whenever $G$ admits an acyclic $k$-coloring. By
Proposition~\ref{prop:ZSk} and
Corollary~\ref{cor:bound_signed}(\ref{cor:bound_signed-1}), we get the
following new upper bound that improves Theorem~\ref{thm:nrs-planar}.

\begin{theorem}\label{thm:main-acyclic}
Let $G$ be a graph that admits an acyclic $k$-coloring. We have
$\chi_s(G)\le k\cdot 2^{k-2}$. 
\end{theorem}

Recall that Huemer et al.~\cite{hfmm08} proved that
Theorem~\ref{thm:alon-marshall} is tight, i.e. there exists a planar
graph $(G)$ such that $\chi_2(G) = k\cdot 2^{k-1}$. By
Lemma~\ref{lem:signed-signified}, we can deduce that $\chi_s(G) \ge
k\cdot 2^{k-2}$, showing that Theorem~\ref{thm:main-acyclic} is actually
tight.

We also get the following upper bounds:

\begin{theorem}\ 
 \begin{enumerate}\label{thm:main}
 \item $\chi_s(\mathcal{O}_4)\le 4$.\label{thm:outer3}
 \item $\chi_s(\mathcal{P}_3)\le 40$.\label{thm:main3}
 \item $\chi_s(\mathcal{P}_4)\le 25$.\label{thm:main4}
 \item $\chi_s(\mathcal{P}_5)\le 10$.\label{thm:main5}
 \item $\chi_s(\mathcal{P}_6)\le 6$.\label{thm:main6}
 \end{enumerate}
\end{theorem}

\begin{proof}
  \begin{enumerate}
  \item The result follows from Theorem~\ref{thm:outer-signified}
    and Corollary~\ref{cor:bound_signed}(\ref{cor:bound_signed-1}).
  \item By Theorem~\ref{thm:alon-marshall} and Borodin's
    result~\cite{b79}, every planar graph admits a signified
    homomorphism to $ZS_5$, a graph on 80 vertices, which is anti-twinned by
    Proposition~\ref{prop:ZSk}. The result then follows from
    Corollary~\ref{cor:bound_signed}(\ref{cor:bound_signed-1}). 
  \item The result follows from Theorem~\ref{thm:main-signified}
    and Corollary~\ref{cor:bound_signed}(\ref{cor:bound_signed-1}).
  \item The result follows from
    Theorem~\ref{thm:main-montejano}(\ref{thm:main5-montejano})
    and
    Corollary~\ref{cor:bound_signed}(\ref{cor:bound_signed-1}).
  \item The result follows from
    Theorem~\ref{thm:main-montejano}(\ref{thm:main6-montejano}) and
    Corollary~\ref{cor:bound_signed}(\ref{cor:bound_signed-1}).
  \end{enumerate}
\end{proof}

Concerning lower bounds, Naserasr et al.~\cite{nrs12} constructed a
planar graph with signed chromatic number $10$. Note that this result
also follows from Theorem~\ref{thm:lower-planar} and
Corollary~\ref{cor:bound_signed}(\ref{cor:bound_signed-1}). If $10$ is
the tight bound for signed chromatic number of planar graphs, we then
get the following from Theorem~\ref{thm:planar-20} and
Lemma~\ref{lem:signed-signified}: 

\begin{theorem}\label{thm:planar-10}
  If every planar graphs admits a signed homomorphism to a
  graph $H_{10}$ of order $10$, then $H_{10}$ is isomorphic to $SP_9^+$.
\end{theorem}

We can easily construct a planar graph of girth $4$ with signed
chromatic number $6$ (see Figure~\ref{fig:borneinf4}). Finally, for
higher girths, note that every even cycle with exactly one negative
edge needs $4$ colors for any signed coloring.

\section{Conclusion}\label{sec:conclusion}

One of our aims was to introduce and study some relevant target graphs
for signified homomorphisms. We studied the anti-twinned graph
$AT(G,\Sigma)$, the signified Zielonka graph $ZS_k$, the signified
Paley graph $\spq$, and the signified Tromp Paley graph $\tr$.
Theorems~\ref{thm:outerplanar},~\ref{thm:planar-20}
and~\ref{thm:planar-10} suggest that such target graphs are indeed
significant.

We proved that there exist planar graphs with signified chromatic
number $20$ and ask whether this bound is tight:

\begin{openprob}
  Does every planar graph admit a signified homomorphism to $\trk{9}$ ?
\end{openprob}

We have checked by computer that every $4$-connected planar
triangulation with at most $15$ vertices admits a homomorphism to
$\trk{9}$. The restriction to $4$-connected triangulations (i.e.
triangulations without separating triangles) is justified by
Lemma~\ref{lem:triangle-transitif}. For the $2^{25}$ non-equivalent
signatures of each of the $6244$ $4$-connected planar triangulations
with $15$ vertices, our computer check took 150 CPU-days. Checking
$4$-connected triangulations with more vertices would require too much
computing power.

Finally, it would be nice to drop the condition ``anti-twinned" in
Theorem~\ref{thm:planar-20}.

\begin{figure}
  \begin{center}
    \includegraphics[scale=0.8]{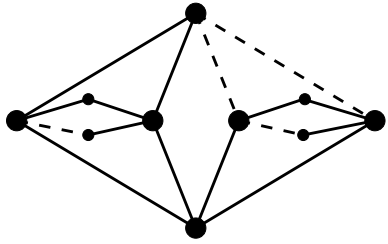}
    \caption{A planar graph of girth 4 with signed chromatic number
      6.\label{fig:borneinf4}} 
  \end{center}
\end{figure}

\end{document}